\documentclass[nofootinbib,aps,prx,amsfonts,amsmath,amssymb,superscriptaddress, floatfix]{revtex4-2}
\usepackage{graphicx}
\usepackage{subcaption}
\usepackage{caption}
\usepackage{overpic}
\usepackage{physics}
\usepackage{amsthm}
\usepackage{mathrsfs}
\usepackage{mathtools}
\usepackage{enumerate}
\usepackage{xcolor}
\usepackage{hyperref}
\hypersetup{
    colorlinks,
    linkcolor=blue,
    citecolor=blue,
    urlcolor=blue
}
\usepackage[ruled]{algorithm2e}
\usepackage{adjustbox}
\usepackage{float}
\usepackage{wasysym}
\usepackage{url}
\usepackage{bm}
\usepackage{stmaryrd}

\usepackage{enumitem}
\setlist[enumerate,1]{label={(\alph*)}}

\newtheorem{theorem}{Theorem}[section]
\newtheorem{definition}[theorem]{Definition}
\newtheorem{lemma}[theorem]{Lemma}

\newtheorem{remark}[theorem]{Remark}
\newtheorem{corollary}[theorem]{Corollary}

\def\orcid#1{\kern .08em\href{https://orcid.org/#1}{\includegraphics[keepaspectratio,width=0.7em]{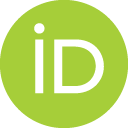}}} 

\newcommand{\RV}[1]{\bm{\mathcal{#1}}}
\newcommand{\expect}[1]{\mathbb{E}\left[ #1 \right]}
\DeclareMathOperator{\Var}{Var}
\newcommand{\measure}[3]{#1 \sim \left( #2, #3\right)}
\newcommand{\prob}[1]{\mathbb{P}\left[ #1 \right]}
\newcommand{\probM}[4]{\ensuremath{ {\mathop{\mathbb{P}}_{\measure{#1}{#2}{#3}}}{\left[#4\right]}}}

\newcommand{\tanft}[1]{\ensuremath{\tan \frac{#1}{2}}}
\newcommand{\tantft}[1]{\ensuremath{\tan^2 \frac{#1}{2}}}
\newcommand{\tannft}[2][1]{\ensuremath{\tan^{#1} \frac{#2}{2}}}
\newcommand{\truncate}[2]{\left\llbracket #1 \right\rrbracket_{#2}}
\newcommand{\normtwo}[1]{\left\lVert #1 \right\rVert_2}
\newcommand{\ceil}[1]{\left\lceil #1 \right\rceil}

\usepackage{tikz}
\usepackage{graphicx} 

\begin{document}
\preprint{APS/123-QED}

\title{More-efficient Quantum Multivariate Mean Value Estimator from Generalized Grover Operator}
\author{Letian Tang\orcid{0000-0001-8882-9187}}
\email{letiant@andrew.cmu.edu}
\affiliation{
Carnegie Mellon University, Pittsburgh, PA 15213, United States
}


\date{\today}

\begin{abstract}
In this work, we present an efficient algorithm for multivariate mean value estimation. Our algorithm outperforms previous work by polylog factors and nearly saturates the known lower bound. More formally, given a random vector $\vec{\bm{\mathcal{X}}}$ of dimension $d$, we find an algorithm that uses $O\left(n \log \frac{d}{\delta}\right)$ samples to find a mean estimate that $\vec{\tilde{\mu}}$ that differs from the true mean $\vec{\mu}$ by $\frac{\sqrt{\tr \Sigma}}{n}$ in $\ell^\infty$ norm and hence $\frac{\sqrt{d \tr \Sigma}}{n}$ in $\ell^2$ norm, where $\Sigma$ is the covariance matrix of the components of the random vector. We also presented another algorithm that uses smaller memory but costs an extra $d^\frac{1}{4}$ in complexity. Consider the Grover operator, the unitary operator used in Grover's algorithm. It contains an oracle that uses a $\pm 1$ phase for each candidate for the search space. Previous work has demonstrated that when we substitute the oracle in Grover operator with generic phases, it ended up being a good mean value estimator in some mathematical notion. We used this idea to build our algorithm. Our result remains not exactly optimal due to a $\log \frac{d}{\delta}$ term in our complexity, as opposed to something nicer such as $\log \frac{1}{\delta}$; This comes from the phase estimation primitive in our algorithm. So far, this primitive is the only major known method to tackle the problem, and moving beyond this idea seems hard. Our results demonstrates that the methodology with generalized Grover operator can be used develop the optimal algorithm without polylog overhead for different tasks relating to mean value estimation.
\end{abstract}

\maketitle

\section{Introduction}

The problem of Mean Value Estimation has been one of the most classic problems in statistics. For example, it has been widely used in Monte Carlo simulations which is useful for topics such as statistical physics \cite{Landau_2021}. With accesses to a quantum computer, it has been shown that, in terms of accuracy, one can estimate the mean quadratically more accurate than with a classical computing device \cite{Kothari_2022}. A natural generalization comes when one wish to instead estimate the mean of a random vector as opposed to a single variable. In the classical setting, there have been a large volume work for multivariate mean estimators \cite{Lugosi_2019, Lugosi_2020}. In the quantum setting, this problem have mainly be tackled by Ref.~\cite{Cornelissen_2022}, where they find the following result:
\begin{theorem}[Previous best Multivariate Estimator, Rephrased]
\label{thm: previous_best}
Given quantum experiments for a $d$-dimensional multivariate random variable $\vec{\RV{X}}$ with unknown covariance matrix $\Sigma$, there is an efficient algorithm that outputs a mean estimate $\vec{\tilde{\mu}}$ such that:
\begin{equation}
    \prob{\norm{\vec{\tilde{\mu}} - \expect{\vec{\RV{X}}}}_\infty \leqslant \frac{\sqrt{\tr \Sigma}}{n}} \geqslant 1 - \delta
\end{equation}
which implies 
\begin{equation}
    \prob{\norm{\vec{\tilde{\mu}} - \expect{\vec{\RV{X}}}}_2 \leqslant \frac{\sqrt{d \tr \Sigma}}{n}} \geqslant 1 - \delta
\end{equation}
The algorithm takes $O\left(n\log \frac{d}{\delta}\operatorname{polylog}(n,d)\right)$ accesses to the quantum experiment.
\end{theorem}

The authors further proved that the algorithm is near-optimal in the sense that there are always some multivariate random variable such that we will have to take $\Omega(n)$ time to achieve accuracy $\norm{\vec{\tilde{\mu}} - \expect{\vec{\RV{X}}}}_2 \leqslant \frac{\sqrt{d \tr \Sigma}}{n}$ with success probability at least $\frac{2}{3}$. In addition to poly-logarithmic improvements, it also leaves open the possibility that alternative metrics or additional structural assumptions on $\vec{\RV{X}}$ could yield further improvements in specific settings. Another related work in Ref.~\cite{Huggins_2022} utilizes the same ideas but discusses the difficult situation where we are only able to simulate a set of observables (such as a frustrated Hamiltonian) with no knowledge of the eigenstates. In that case we are only able to achieve mean (expectation) estimation where the eigenvalues are bounded. 

Recently, in Ref.~\cite{Kothari_2022}, cleverly utilizing a generalization of the Grover operator, the authors were able to provide a univariate estimator that gets rid of all log factors in algorithm design. This paves the way for improvements in the multivariate case. Employing this idea, in this paper, we are able to find:
\begin{theorem}[Main Result for Meticulous Estimator]
\label{thm: mer_main}
Given quantum experiments for a $d$-dimensional multivariate random variable $\vec{\RV{X}}$ with unknown covariance matrix $\Sigma$, there is an efficient algorithm that outputs a mean estimate $\vec{\tilde{\mu}}$ such that:
\begin{equation}
    \prob{\norm{\vec{\tilde{\mu}} - \expect{\vec{\RV{X}}}}_\infty \leqslant \frac{\sqrt{\tr \Sigma}}{n}} \geqslant 1 - \delta
\end{equation}
which implies 
\begin{equation}
    \prob{\norm{\vec{\tilde{\mu}} - \expect{\vec{\RV{X}}}}_2 \leqslant \frac{\sqrt{d \tr \Sigma}}{n}} \geqslant 1 - \delta
\end{equation}
The algorithm takes
\begin{itemize}
    \item $O\left(n \log \frac{d}{\delta}\right)$ in terms accesses to the quantum experiment
    \item $O\left(\log n \log \log n\right)$ in terms of quantum registers needed to hold quantum experiments.
\end{itemize}
\end{theorem}

This removes almost all log factors but the $\log \frac{d}{\delta}$ factor instead of something nicer such as $\log \frac{1}{\delta}$. Removing this factor is very difficult based on our current knowledge, as will be discussed in Sec.~\ref{sec: discussion}.

In practice, our quantum computer might be extremely limited in memory, so we provide a simpler, memory-efficient algorithm, accepting an extra $d^\frac{1}{4}$ factor in time complexity:
\begin{theorem}[Main Result for Simple Estimator]
\label{thm: sim_main}
Given quantum experiments for a $d$-dimensional multivariate random variable $\vec{\RV{X}}$ with unknown covariance matrix $\Sigma$ , there is an efficient algorithm that outputs a mean estimate $\vec{\tilde{\mu}}$ such that:
\begin{equation}
    \prob{\norm{\vec{\tilde{\mu}} - \expect{\vec{\RV{X}}}}_\infty \leqslant \frac{\sqrt{\tr \Sigma}}{n}} \geqslant 1 - \delta
\end{equation}
which implies 
\begin{equation}
    \prob{\norm{\vec{\tilde{\mu}} - \expect{\vec{\RV{X}}}}_2 \leqslant \frac{\sqrt{d \tr \Sigma}}{n}} \geqslant 1 - \delta
\end{equation}
The algorithm takes
\begin{itemize}
    \item $O\left(n d^\frac{1}{4} \log \frac{d}{\delta}\right)$ in terms accesses to the quantum experiment
    \item $O\left(1\right)$ in terms of quantum registers needed to hold quantum experiments.
\end{itemize}
\end{theorem}
This algorithm still have the potential to outperform the existing best estimator in the situation where $d$ is not too big and we would like to know the accuracy up to a very high precision. 

After this introduction, this paper will be split into three sections. First, in Sec.~\ref{sec: preliminaries} we discuss all the necessary definitions and primitives needed for this paper; In Sec~\ref{sec: key_subroutine} we re-analyze the Grover operator and proved a different property than discussed in Ref.~\cite{Kothari_2022}. We then refine these properties and eventually build a quite straight-forward univariate estimator. In Sec.~\ref{sec: multi_estimator} we expand our ideas to the multivariate case and construct two different estimators for the case where an upper bound of $\tr \Sigma$, where $\Sigma$ is the covariance matrix, is known. Lastly, we upgrade to handle an unknown $\tr \Sigma$ and reach our conclusions. We discuss further details and future prospects in Sec.~\ref{sec: discussion}.

\section{Preliminaries}
\label{sec: preliminaries}

This section outlines the basic setups for this paper, including notations and useful primitives. 

\subsection{Notations and Definitions}

First we define some basic concepts and notations for mathematical rigor. 
\begin{definition}[Finite Probability Space]
    \label{def:probaiblity_space}
    A finite probability space is a pair $(\Omega, p)$ where $\Omega$ is a finite set of possible outcomes and $p: \Omega \to [0,1]$ is a probability mass function satisfying $\sum_{k \in \Omega} p(k) = 1$, we can abbreviate $p(k)$ as $p_k$. \footnote{Most authors would probably use $\omega$ to indicate an element from probability space $\Omega$. Here we use $k$, an index variable. The exact mathematical object as elements in $\Omega$ does not matter, so in practice we could just label them with some index, that is, $\Omega = \{\omega_{k}\}_{i \in \left[\left|\Omega\right|\right]}$. Then we might just represent these elements $\omega_k$ with the index $k$. More rigorously, while $k \in \Omega$ is not necessarily an integer index, we can effectively think of it as an index. }
\end{definition} 

\begin{definition}[Hilbert Space]
    \label{def:hilbert_space}
    Given some finite set $X$, the Hilbert space on $X$ is noted as $\mathcal{H}_X$, which has $\{\ket{x}\}_{x \in X}$ as an orthornormal basis. 
\end{definition}
\begin{remark}[Labels in Kets]
    Within this paper, whenever we designate a variable $x$ from a set $X$, the vectors labeled with this variable $\{\ket{x}\}_{x \in X}$ automatically denotes the basis for $\mathcal{H}_X$ as clarified in Definition~\ref{def:hilbert_space}. In general, we might like to use a more explicit notation such as $\ket{x}_X$ to avoid ambiguity, but within the context of this paper, dropping the $X$ subscript suffice. 
\end{remark}

\begin{definition}[Synthesizer]
    Let $(\Omega, p)$ be a probability distribution. A synthesizer is a unitary operator $\mathcal{P}$ acting on $\mathcal{H}_\Omega \otimes \mathcal{H}_{\mathrm{anc}}$:
    \begin{equation}
\mathcal{P} \ket{0} = \sum_{k \in \Omega} \sqrt{p_k}\ket{k}
\end{equation}
    where $\ket{\mathrm{aux_k}} \in \mathcal{H}_{\mathrm{anc}}$ are normalized vectors, and $\ket{0}$ denotes some state that is easily achievable.
\end{definition}

\begin{remark}[Potential Auxiliaries]
    \label{remark: ancilla}
    In general one might write 
    \begin{equation}
\mathcal{P} \ket{0} = \sum_{k \in \Omega} \sqrt{p_k}\ket{k}\ket{\mathrm{aux}_k}
\end{equation}
    where $\mathcal{H}_{\mathrm{anc}}$ is the Hilbert space of some ancilla qubits that may encode extra information. For the scope of this paper we will ignore ancillas for simplicity and clarity. (One can always recover ancillas by either re-inserting them in appropriate locations throughout derivation process, or by redefining the states $\ket{k}$ to incorporate potential ancillas.)
\end{remark}

Here are definitions relating to random variables.
\begin{definition}[Random Variable]
    A univariate (real) variable $\bm{\mathcal{X}}$ on a finite probibility space $(\Omega, p)$ is a mapping $\Omega \to \mathbb{R}$. For each $k \in \Omega$, we can abbreviate $\bm{\mathcal{X}}(k)$ as $\mathcal{X}_k$, with bold font and nonbold font accordingly. 

    Similarly, we can define a univariate complex random variable $\bm{\mathcal{Y}}$ as mapping $\Omega \to \mathbb{C}$ and introduce abbreviations $\mathcal{Y}_k$ acoordingly. 
\end{definition}

\begin{remark}[Notation]
    We will use bold font (and usually calligraphic) to represent random variables. 
\end{remark}

\begin{remark}[Omitting $\Omega$]
    In our notations we often omit $\Omega$ because in the scope of this paper we always fix a single probability space. 
\end{remark}

There are some special univariate random variables that will be useful in the future.
\begin{definition}(Identity Random variable)
    In particular, $\bm{1}$ is the random variable $\Omega \to \{1\}$.
\end{definition}

\begin{definition}(Indicator Random variable)
    Given some event $A$, $\bm{1}_A$ is the indicator random variable that is 1 when $A$ is true and 0 otherwise, that is, for all $k \in \Omega$
    \begin{equation}
        (\bm{1}_A)_k = \begin{cases}
            1 & \text{if } A \\
            0 & \text{if } \neg A \\
        \end{cases}
    \end{equation}
\end{definition}

It is convenient to label quantum states with random variables which we will define as follows:
\begin{definition}[Associated Quantum State]
    Given a complex univariate random variable $\RV{X}$ in finite probaility space $(\Omega, p)$, its associated quantum state, denoted $\ket{\bm{\mathcal{X}}}$ is a state in $\mathcal{H}_\Omega$ with
    \begin{equation}
\ket{\bm{\mathcal{X}}} = \sum_{k \in \Omega} \sqrt{p_k} \mathcal{X}_k \ket{k} 
\end{equation}
\end{definition}

\begin{remark}
    For synthesizer $\mathcal{P}$ we have
    \begin{equation}
\mathcal{P}\ket{0} = \ket{\bm{1}}
\end{equation}
\end{remark}

\begin{remark}
We have $\bra{k}\ket{\RV{X}} = \sqrt{p_k}\RV{X}_k$ for $k \in \Omega$.
\end{remark}

Our ultimate goal is to investigate multivariate random variables defined as:
\begin{definition}[Multivariate Random Variable]
    A multivariate (real) random variable $\vec{\bm{\mathcal{X}}}$ of dimension $d$ is a $d$-dimension tuple of (real) random variables $\vec{\bm{\mathcal{X}}} = (\bm{\mathcal{X}}^1, \bm{\mathcal{X}}^2, \cdots, \bm{\mathcal{X}}^d)$, each from probability space $(\Omega, p)$. As suggested, for $\alpha \in [d]$ we would denote the $\alpha$-th element with superscript $\bm{\mathcal{X}}^\alpha$, and an instance for $k \in \Omega$ of that element $\mathcal{X}^\alpha_k$.
    
    Similarly, we can define a multivariate complex random variable and use notation accordingly. There is also the random variable $\vec{\bm{1}} = (\bm{1}, \bm{1}, \cdots, \bm{1})$.

    Lastly, for $k \in \Omega$ an instance of a mulvariate random variable is a vector $\vec{\mathcal{X}}_k = (\mathcal{X}_k^1, \mathcal{X}_k^2, \cdots, \mathcal{X}_k^d)$.
\end{definition}

Drawn from physicists' playbook, we will utilize both superscripts and subscripts for different indices.
\begin{remark}[Index]
    We will consistently use Greek letters as dimension indices in $[d]$ and put them as superscripts; conversely, we will use Latin letters as elements from $\Omega$ and put them as subscripts.
\end{remark}


Here we introduce notations on operations between random variables:
\begin{definition}[Operations on Random Variables]
    Whenever we have an elementary operation (addition, multiplication, exponential, trig function, etc) acting on one or two random variables, the outcome is the element-wise result of such operation on the random variables, with the same dimensions. 
\end{definition}

\begin{remark}[Example]
    As examples, we have:
    \begin{itemize}
        \item $\RV{C} = \RV{A} + \RV{B}$ will be a univariate random variable with $\mathcal{C}_k = \mathcal{A}_k + \mathcal{B}_k$ for all outcome $k \in \Omega$.
        \item $\vec{\RV{B}} = e^{\vec{\RV{A}}}$ is a univariate random variable with $\mathcal{B}_k^\alpha = e^{\mathcal{A}_k^\alpha}$, for $k \in \Omega$ and $\alpha \in [d]$.
        \item $\RV{B} = \arcsin \RV{A}$ is a random variable with $\RV{B}_k = \arcsin \RV{A}_k$ for $k \in \Omega$.
    \end{itemize}
\end{remark}

There are some miscellaneous operations and definitions that will be important in the future.
\begin{definition}[Truncation]
    \label{def: truncation}
    Consider a univariate random variable $\RV{X}$, we denote its truncation to threshold $K$ (where $K \geqslant 0$) $\truncate{\RV{X}}{K}$ to be a random variable such that for all $k \in \Omega$:
    \begin{equation}
        \left(\truncate{\RV{X}}{K}\right)_k = \begin{cases}
        -K & \mathcal{X}_k < -K \\
        \mathcal{X}_k & \left| \mathcal{X}_k \right| \leqslant K \\
        K & \mathcal{X} > K \\
        \end{cases} 
    \end{equation}
\end{definition}
We will also need to define truncation for multivariate random variable.
\begin{definition}[Truncation of Multivariate Random Variable]
    \label{def: truncation}
    Consider a multivariate random variable $\vec{\RV{X}}$, we denote its truncation to threshold $K$ (where $K \geqslant 0$) $\truncate{\RV{X}}{K}$ to be a random variable such that for all $k \in \Omega$:
    \begin{equation}
        \left(\truncate{\RV{X}}{K}\right)_k = \begin{cases}
        \vec{\mathcal{X}_k} & \normtwo{\vec{\mathcal{X}_k}} \leqslant K \\
        0 & \text{otherwise}
        \end{cases}
    \end{equation}
\end{definition}

\begin{definition}[Inner Product]
    Given vectors $\vec u, \vec v \in \mathbb{R}^d$, there inner product $\langle u, v \rangle$ is defined as $\langle u, v \rangle = \sum_{\alpha=1}^{d} u^\alpha v^\alpha$, where subscript is the index for dimension.

    Given vector $\vec u \in \mathbb{R}^d$ and read multivariate random variable $\vec{\RV{X}}$, their inner product is a univariate random variable $\langle \vec u, \vec{\RV{X}} \rangle = \sum_{\alpha = 1}^{d} u^\alpha \RV{X}^\alpha$.
\end{definition}

\begin{definition}[Vector Encoding]
    Given vector $\vec u \in S \subseteq \mathbb{R}^d$, where $S$ is some finite set that is made clear within the context, we can use a quantum register in Hilbert space $\mathcal{H}_S$ to encode the vector denoted as $\ket{\vec u}$.
\end{definition}

A random variable has expectations and variances which we will define here:
\begin{definition}[Expectation]
    Given a univariate random variable $\RV{X}$, its expectation $\expect{\RV{X}}$ is defined as $\expect{\RV{X}} = \sum_{k \in \Omega} p_k \mathcal{X}_i$.

    For a multivariate random variable with dimension $d$ $\vec{\RV{X}}$, its expectation $\mathbb{E}[\vec{\RV{X}}]$ is a vector with $\expect{\vec{\RV{X}}} ^ \alpha = \expect{\RV{X}^\alpha}$ for $\alpha \in [d]$.
\end{definition}

\begin{remark}
For univariate random variables $\RV{X}$, $\RV{Y}$, $\bra{\RV{X}}\ket{\RV{Y}} = \expect{\RV{X}^* \RV{Y}}$.
\end{remark}

\begin{definition}[Variance]
    For a univariate random variable $\RV{X}$, its variance is $\sigma^2 = \expect{\left(\RV{X} - \expect{\RV{X}}\right)^2}$. The variance can also be denoted $\Var \RV{X}$.

    For a multivariate random variable $\vec{\RV{Y}}$, its covariance matrix $\Sigma$ satisfies $\Sigma^{\alpha\beta} = \expect{\left(\RV{Y}^\alpha - \expect{\RV{Y}^\alpha}\right)\left(\RV{Y}^\beta - \expect{\RV{Y}^\beta}\right)}$
\end{definition}

There are also some basic concepts from quantum information that needs to be addressed and assigned with notations for convenience.
\begin{definition}[Measurement]
As in quantum mechanics, in some Hilbert space $\mathcal{H}$, given an observable (Hermitian operator) $\mathcal{A}$ and a normalized state $\ket{\psi}$, measuring the observable on the state $\ket{\psi}$ generates a probability space and a random variable (as the resulting eigenvalues). Let the random variable be $\RV{\lambda}$, we denote this as $\measure{\RV{\lambda}}{\mathcal{A}}{\ket{\psi}}$, and the probability of some event $E$ depending on $\RV{\lambda}$ being $\probM{\RV{\lambda}}{\mathcal{A}}{\ket{\psi}}{E(\RV{\lambda})}$.
\end{definition}
\begin{remark}
When we are describing probabilities resulting from some algorithm, the labels for the Hilbert space, the subscript after $\mathbb{P}$ can be dropped.
\end{remark}

\begin{definition}[Phases of Unitary]
Given some unitary $\mathcal{U}$, let $\arg(\mathcal{U})$ be a Hermitian operator with eigenvalues in $(-\pi, \pi]$ such that $\mathcal{U} = \exp{i \arg (\mathcal{U})}$. In other words, $\arg(\mathcal{U})$ will denote the operator corresponding to the phases of $\mathcal{U}$.
\end{definition}

\begin{remark}[Phase Estimation]
Given a unitary operator $\mathcal{U}$ and input state $\ket{\psi}$, as an example, phase estimation aims to produce a random variable $\measure{\RV{\theta}}{\arg(\mathcal{U})}{\ket{\psi}}$.
\end{remark}

\begin{definition}[absolute value mod $2\pi$]
    Given a real number (in practice, a phase) $\varphi$, we define $|\varphi|_{2\pi}$ as $\min_{n \in \mathbb{Z}} |\varphi - 2\pi n|$.
\end{definition}

\subsection{Primitives}

\label{sec: primitives}

In this section, we set up the problem of mean value estimation and provide useful algorithmic tools. First, we introduce the notion of a quantum experiment.

\begin{definition}[Quantum Experiment]
    \label{def: experiment}
    Given a multivariate random variable $\vec{\RV{X}}$ of dimension $d$, a quantum experiment is a unitary $U$ such that for all $k \in \Omega$
    \begin{equation}
U \ket{k}\ket{0} = \ket{k}\ket{\vec{\mathcal{X}_k}}
\end{equation}
    The unitary acts on $\mathcal{H}_\Omega \otimes \mathcal{H}_{\mathrm{vec}}$, where $\mathcal{H}_{\mathrm{vec}}$ is the Hilbert space with basis representing quantum states corresponding to vectors in $\mathbb{R}^d$.
\end{definition}

\begin{remark}[Incorporating Ancilla]
The experiment can be redefined to incorporate extra changes in other ancillas as
\begin{equation}
U \ket{k}\ket{0}\ket{0} = \ket{k}\ket{\vec{\mathcal{X}_k}}\ket{\mathrm{aux'}_k}
\end{equation}
Just as in Remark~\ref{remark: ancilla} we drop it for simplicity.
\end{remark}

In practice, we might not have access to either the synthesizer or the quantum experiment, but presented with access to a complete quantum experiment in the form of the following
\begin{definition}[Complete Quantum Experiment]
\label{def: complete_quantum_exp}
Given a multivariate random variable $\vec{\RV{X}}$, a complete quantum experiment is a unitary $V$ with
\begin{equation}
V \ket{0}\ket{0} = \sum_{k \in \Omega} \sqrt{p_k}\ket{k}\ket{\vec{\mathcal{X}_k}}
\end{equation}
\end{definition}

One can easily see that $V = U \mathcal{P}$ where $U$ is the quantum experiment for $\vec{\RV{X}}$. As will eventually be discussed in Remark~\ref{remark: join}, access to controlled-$V$ and controlled-$V^\dag$, in the end, is the only thing that we need. 

The quantum experiment highlighted here can mean a number of things. Consider a classical Monte Carlo process in some simulation code, we can transport all parts of the code into a quantum computer and the resulting unitary will be $U$. From a complexity perspective, this draws an even comparison with a quantum and a classical mean estimator. In addition, this language generalizes to a quantum sensor, where the quantum computer, via the unitary, is coupled to the physical world. \footnote{Of course, it needs to be hooked up to a fault-tolorant device to run our algorithm.} With some modifications we can further expand this language for inherently quantum information processes without one-to-one classical analogues, such as quantum chromodynamics simulations \cite{Atas_2023, preskill_2018} or variational quantum circuits used in QML \cite{Wang_2024}. Further intricacies of these discussions as will be explored in Sec~\ref{sec: discussion}. 




Similarly, for a univariate random variable, there is also a notion of quantum experiment implemented below:
\begin{definition}[Univariate Quantum Experiment]
    \label{def: univariate experiment}
    Given a univariate random variable $\RV{X}$, a quantum experiment is a unitary $U$ such that for all $k \in \Omega$
    \begin{equation}
U \ket{\Omega}\ket{k}\ket{0} = \ket{k}\ket{\mathcal{X}_k}\ket{\mathrm{aux}'_k}
\end{equation}
    The unitary acts on $\mathcal{H}_\Omega \otimes \mathcal{H}_{\mathrm{vec}} \otimes \mathcal{H}_{\mathrm{anx'}}$, where $\mathcal{H}_{\mathrm{vec}}$ is the Hilbert space with basis representing quantum states corresponding to vectors in $\mathbb{R}^d$.
\end{definition}

One can then easily define complete univariate quantum experiment, which we will omit. 

Note that:
\begin{remark}[Drawing Classical Samples from Quantum Experiment]
\label{remark: quantum_ex_to_classical}
Using one access to the (complete) quantum experiment for a random variable (either univairate or multivariate) $\RV{X}$ ensured by a measurement, one can draw a sample of $\RV{X}$ in the classical sense. 
\end{remark}

Note that given a fault-tolerant quantum computer, we can make simple transformations of the outcome of a random experiment, and it is effectly transforming the random variable of concern.
\begin{remark}[Post-processing Quantum Experiment]
   \label{remark: quantum_experiment_algebra}
   Given a simple function $f$ that is efficiently computable and a univariate quantum experiment of some random variable $\mathcal{X}$, we can transport the computation process to a quantum computer and construct a quantum experiment for $f(\mathcal{X})$ with corresponding small overhead. 
\end{remark}

Given access to a quantum experiment, we can build a phase oracle that will be crucial to our mean value estimation algorithm:

\begin{definition}[Univariate Phase Oracle]
    \label{def: univariate phase oracle}
     With $O(1)$ call to the quantum experiment that implements that implements some random variable $\RV{X}$ (and its inverse), we can construct a phase oracle $\mathcal{O}$ on $\mathcal{H}_S \otimes \mathcal{H}_\Omega$. For $\vec u \in S$ and $k \in \Omega$, where $S$ is some set of vectors in $\mathbb{R}^d$, we have
    \begin{equation}
\mathcal{O}\ket{k} = e^{i \mathcal{X}_k} \ket{k}
\end{equation}
\end{definition}
Overall, many of the above definitions follow those from Ref.~\cite{Kothari_2022}. In this paper, we will ignore the potential errors and infidelities during implementations of the primitives introduced in this paper. A more meticulous reader may reintegrate these errors back into our analysis. 

The phase estimation algorithm is one of the most classic algorithm of quantum computing, and we will briefly describe it as:
\begin{theorem}[Phase Estimation]
\label{thm: uni_phase_estimation}
Given a unitary $U$ on some register in $\mathcal{H}$. Let $U$ to have some eigenvector $\ket{\varphi}$ with eigenvalue $e^{2\pi i \varphi}$. We set $\varphi \in (-\frac{1}{2}, \frac{1}{2}]$.

Say we initialize the register $\mathcal{H}$ with some easy-to-prepare state $\ket{\psi}$. Using $N$ access to controlled versions of $U$ (entangling with an extra register to control), there is an algorithm that returns an estimate $\tilde{\varphi}$ such that: 
\begin{equation}
    \prob{\left|\tilde{\varphi}-\varphi\right| \leqslant \frac{\kappa}{N}} \geqslant \left|\bra{\varphi}\ket{\psi}\right|^2\left(1-\frac{1}{2(\kappa-1)}\right)
\end{equation}
where $\kappa \in \mathbb{N}^+$, and note that $\left|\bra{\varphi}\ket{\psi}\right|^2 = \probM{\RV{\theta}}{\arg U}{\ket{\psi}}{\RV{\theta} = 2 \pi \varphi}$.
\end{theorem}

More generally, we can extend the phase estimation into the multivariate case, as an algorithm that we will call ``multidimensional phase estimation.'' It is originally used to estimate the gradient of some function proposed in Refs.~\cite{Jordan_2005, bulger_2005}. Later authors polished the results and referred to it as Jordan's algorithm in Ref.~\cite{Gilyen_2019}. We, in a similar fashion as in Ref.~\cite{Cornelissen_2022}, will reformulate results mainly following Ref.~\cite{Gilyen_2019}.

First, we consider vectors from the following set:
\begin{definition}[Hypercubic Lattice]
For a dimension $d$ and resolution $N$ (which is usually a power of 2), define the hypercubic lattice to be:
\begin{equation}
    G_N = \left\{ j \in \{0, 1, \cdots, N-1\}:  \frac{j}{N} - \frac{1}{2} - \frac{1}{2N} \right\}^d
\end{equation}
\end{definition}

The multidimensional phase estimation requires the following component which we call multivariate phase unitary:
\begin{definition}[Multivariate Phase Unitary]
    Given a vector $\vec x$, a multivariate phase unitary $\mathcal{U}$ acts on $\mathcal{H}_{G_N}$ via:
    \begin{equation}
        U \ket{\vec u} = e^{2 \pi i \expval{\vec u, \vec x}} \ket{u}
    \end{equation}
    where $\vec u \in G_N$, $G$ is the hypercubic lattice for dimension $d$ and resolution $N$. 
\end{definition}

\begin{algorithm}
\caption{\label{alg: phase_estimation} Multidimensional Phase Estimation}
\KwData{Resolution $N$, dimension $d$, Multivariate Phase Unitary $\mathcal{U}$ on $\mathcal{H}_{G_N}$}
\KwResult{Some estimate of the phase $\vec y$}
Initialize state $\ket{\psi} \gets \frac{1}{N^{\frac{d}{2}}} \sum_{\vec u \in G_N} \ket{\vec u}$ \;
Obtain state $\ket{\varphi} \gets U^N \ket{\psi}$ \;
Apply inverse QFT: $\ket{\psi_2} \gets 
\mathsf{QFT}_{G_N}^{-1}\ket{\psi_1}$ \;
Measure all dimensions and obtain the estimate $\vec y$.
\end{algorithm}

Now, algorithm~\ref{alg: phase_estimation} illustrates the process to recover the multivariate phase parameter $\vec x$ from the unitary. The $\mathsf{QFT}_{G_N}$ represents Fourier transforming the register that corresponds to each dimension in $\vec u \in G_N$. Its definition can be found in Definition 5.2 (or Definition 17 in preprint) in Ref.~\cite{Gilyen_2019}. The following theorem summarizes the accuracy, which, as written in Ref.~\cite{Gilyen_2019}, is eventually sourced from Ref.~\cite{Nielsen_Chuang_2010}. 

\begin{theorem}[Multidimensional Phase Estimation]
For each dimension $\alpha \in [d]$, the probability for the outcome of Algorithm~\ref{alg: phase_estimation} satisfy:
\begin{equation}
\label{eq: perfect_phase_estimation}
    \mathbb{P} \left[ |x^\alpha - y^\alpha| > \frac{\kappa}{N} \right] \leqslant \frac{1}{2(\kappa-1)} 
\end{equation}
\end{theorem}

In algorithm~\ref{alg: phase_estimation} we assumed that the unitary $U$ is perfectly the multivariate phase unitary. If we are using some other unitary $V$ to approximate $U^N$ there will be another term in error:
\begin{theorem}[Multidimensional Phase Estimation with Noise]
\label{thm: phase_estimation}
If we are using another unitary $V$ in the place of $U^N$ in Algorithm~\ref{alg: phase_estimation}, such that we end up achieving $\ket{\varphi'}$ instead of $\ket{\varphi}$, with $\norm{\ket{\varphi'} - \ket{\varphi}} \leqslant \varepsilon$ for some $\varepsilon$, Eq.~(\ref{eq: perfect_phase_estimation}) is modified to:
\begin{equation}
\label{eq: phase_estimation}
    \mathbb{P} \left[ |x^\alpha - y^\alpha| > \frac{\kappa}{N} \right] \leqslant \frac{1}{2(\kappa-1)} + 2 \varepsilon
\end{equation}
\end{theorem}
The proof for the above theorem is directly refactored from Lemma 5.2 (or Lemma 20 in preprint) in Ref.~\cite{Gilyen_2019} with specific values replaced by symbolic variables. 
\begin{corollary}
    \label{cry: phase_estimation}
    In theorem~\ref{thm: phase_estimation} we might take $\varepsilon = \frac{1}{12}, \kappa = 4$, and we get:
    \begin{equation}
    \mathbb{P} \left[ |x^\alpha - y^\alpha| > \frac{4}{N} \right] \leqslant \frac{1}{3}
\end{equation}
\end{corollary}

Going to hybrid algorithms, we will use the fact that one can easily translate a hybrid circuit into a quantum subroutine.
\begin{theorem}[Hybird to Quantum Conversion]
\label{thm: c_to_q}
Consider a hybrid circuit as a classical deterministic circuit controlling some quantum circuit. The classical circuit can not choose to skip the execution of the quantum circuit, it only feeds in parameters each time the quantum circuit is invoked. The classical circuit calls the quantum circuit $O(g(n, \delta))$ times and costs $O(f(n, \delta))$ by some measure of time complexity, where $0 < \delta < 1$. \footnote{\label{footnote: O_with_delta} For the sake or rigor, whenever we put $\delta$ in big-$O$ notation, we are setting $\delta \to 0$. So our definition of big-$O$ notation with $\delta$ only considers situations where $\delta < c < 1$ where $c$ is some constant. This way, we can say statements such as $O\left(n\ceil{\log \frac{1}{\delta}}\right) = O\left(n\log \frac{1}{\delta}\right)$.}

If the circuit output $\bm{x}$ satisfying some ``successful condition'' $P(\bm{x})$ with probability at least $1 - \delta$. Then by translating the circuit into a classical reversible circuit and then a quantum circuit, one can construct a unitary $U$ utilizing the quantum subrountine $O(g(n, \delta))$ times with total cost $O(f(n, \delta))$. Omitting auxiliaries, It acts on $\mathcal{H}_{\Omega'} \otimes \mathcal{H}_{\text{output}}$ and does the following
\begin{equation}
    U\ket{0}\ket{0} = \sum_{j \in \Omega'} \sqrt{q_j} \ket{\lambda_j}\ket{x_j}
\end{equation}
where $\Omega'$ is the set of outcomes from all measurements during the classical algorithm. Each outcome is defined as a tuple of all outcomes from all the $O(f(n, \delta))$ measurements done at each time the quantum algorithm being run. So $\Omega'$ is another probability space. $\{\ket{\lambda}_j\}$ is a set of orthogonal vectors as corresponding quantum states for each outcome. It lives in $n$ quantum registers. We absorbed potential phases into anxillary states, which we ignore (or can be absorbed into $\ket{\lambda_j}$ via redefinition, similar to Remark~\ref{remark: ancilla}). We can treat $U$ as a complete univariate quantum experiment for some random variable $\bm{x}$ in probability space $\Omega'$ where $\{x_j\}$ are its values. Then:
\begin{equation}
    \mathbb{P}[P(\bm{x})] \geqslant 1 - \delta
\end{equation}
The unitary also needs allocate $O(g(n, \delta))$ quantum registers for the execution for each call to the quantum subroutine. 
\end{theorem}

Then, there are also some well-known algorithmic tricks for classical computing that will be helpful for this work. First, for an algorithm that attempts to locate a value that is within a desriable range, say it success with probability $\frac{2}{3}$. We can arbitrarily boost this probability with repetitions and taking the medium.
\begin{theorem}[Boosting Success Probability]
\label{thm: median}
Consider some algorithm that outputs some value $x$, such that $x \in I$ with probability at least $\frac{2}{3}$, where $I$ is some desirable interval. By repeating the algorithm $ 2 \left\lceil \frac{18 \ln\frac{1}{\delta} - 1}{2} \right\rceil + 1$ (an odd number at least $18 \ln \frac{1}{\delta}$) times and taking the medium, we find an output $y \in I$ with probability at least $1 - \delta$.
\end{theorem}
The above theorem can be easily proven from the fact that the combined algorithm is guaranteed to succeed when at least half of all calls succeeds and Hoeffding's inequality. 

Second, there is a commonly used technique adapted from Ref.~\cite{Kothari_2022}. Consider a subroutine which depends on an accuracy parameter $\varepsilon$ and confidence parameter $\delta$ such that it takes $O\left(\frac{1}{\varepsilon} \log \frac{1}{\delta}\right)$ time to succeed with probability at least $1 - \delta$. There is a way to chain together a sequence of access to the algorithm, with exponentially decaying accuracy parameters, such that the combined cost is $O\left(\frac{1}{\varepsilon}\log(\delta)\right)$ with success probability at least $1 - \delta$, where $\varepsilon$ is the accuracy parameter of the last step. The exact theorem is as follows, with more details added. 
\begin{theorem}[Log log trick]
\label{thm: log_log}
Give some algorithm $\mathcal{A}$ depending on parameters $(\varepsilon, \delta)$, where $\varepsilon > 0$ and $0 < \delta < 1$, such that 
\begin{itemize}
    \item Algorithm always costs $O\left(\frac{1}{\varepsilon} \log \frac{1}{\delta} \right)$ by some measure of complexity and $O\left(\log \frac{1}{\delta}\right)$ by another \footnote{When we use this theorem, the two measures will eventually leads to time and space complexities};
    \item Algorithm ``succeed'' with probability at least $1 - \delta$ whenever it is called. ``succeed'' can be defined as some logical predicate.
\end{itemize}

Fix $\varepsilon, \delta$. Consider calling the algorithm $T$ times. Let the $j$-th time the algorithm to be called with $(\varepsilon'_j, \delta'_j)$. We fix $\varepsilon'_j$ such that $\varepsilon'_T = \varepsilon$ and $\varepsilon'_{j+1} \leqslant \frac{\varepsilon'_{j}}{R}$, where $R > 1$ is some fixed constant. By setting 
\begin{equation}
    \delta'_j = \frac{6}{\pi^2} \frac{1}{\left(T-j + 1\right)^2} \delta
\end{equation}
We can make sure that 
\begin{itemize}
    \item All calls to algorithm $\mathcal{A}$ succeed simultaneously with probability at least $1 - \delta$.
    \item The combined cost is $O\left(\frac{1}{\varepsilon} \log \frac{1}{\delta}\right)$ and $O\left(T \log \frac{T}{\delta}\right)$ by the two measures of complexity respectively. 
\end{itemize} 
\end{theorem}

This algorithm is called the ``log log trick'' because without this technique (i.e., if we fix the $\delta_j$s to be a constant, we will end up getting an extra $\log \log \frac{1}{\varepsilon}$ in our complexity. In Ref.~\cite{Kothari_2022} the authors provided an explicit choice with $\delta'_j = \delta e^{- D \left(\frac{\varepsilon_j}{\varepsilon}\right)^{\frac{1}{2}}}$ where $D = D(R)$ is some constant depending on $R$. However, for the purpose of this paper we also want to reduce another, yet unspecified, measure of complexity. (It will eventually leads to space complexity for the meticulous estimator in Algorithm~\ref{alg: con_mer_estimator} and we'll see why.) For this purpose we need an upgrade to the set of parameters discussed in the theorem statement. The proof to this claim can be found in Appendix~\ref{app: log_log}.

Lastly, There are many algorithms useful for the final classical reduction in Sec.~\ref{sec: final_classical_reduction}. First is the well-known median of means algorithm for classical univariate mean estimation.
\begin{theorem}[Classical Univariate Mean Estimator]
\label{thm: classical_uni_estimator}
There is a classical algorithm (median of means) that given a univariate random variable $\RV{X}$, confidence parameter $\delta$, and $n \in \mathbb{N}^+$, it uses $O\left(n\log \frac{1}{\delta}\right)$ draws to return a mean estimate $\tilde{\mu}$ such that:
\begin{equation}
    \prob{\left|\tilde{\mu} - \expect{\RV{X}}\right| \leqslant \sqrt{\frac{\Var{\RV{X}}}{n}}} \geqslant 1 - \delta
\end{equation}
\end{theorem}
As mentioned, this is a very well-known algorithm, but one can find more discussions on Ref.~\cite{Minsker_2023}.

In the multivariate case, we can also introduce a classical multivariate mean value estimator from Ref.~\cite{Lugosi_2019}.
\begin{theorem}[Classical Multivariate Mean Estimator]
\label{thm: classical_multi_estimator}
There is a classical algorithm that given a multivariate random variable $\vec{\RV{X}}$ of dimension $d$, confidence parameter $\delta$, and $n \in \mathbb{N}^+$, it uses $O\left(n \log \frac{1}{\delta}\right)$ draws to return a mean estimate $\vec{\tilde{\mu}}$ such that:
\begin{equation}
\label{eq: classical_multi_estimator}
\prob{\normtwo{\vec{\tilde{\mu}} - \expect{\vec{\RV{X}}}} \leqslant \sqrt{\frac{\tr \Sigma}{n \ln \frac{1}{\delta}}} + \sqrt{\frac{\normtwo{\Sigma} }{n}}} \geqslant 1 - \delta
\end{equation}
where $\Sigma$ is the covariance matrix. \footnote{As a comment, the original paper format it as the algorithm takes $O(n)$ such that:
\begin{equation}
\prob{\normtwo{\vec{\tilde{\mu}} - \expect{\vec{\RV{X}}}} \leqslant \sqrt{\frac{\tr \Sigma}{n}} + \sqrt{\frac{\normtwo{\Sigma} \ln \frac{1}{\delta}}{n}}} \geqslant 1 - \delta
\end{equation}
We then substituted $n \to n \ln \frac{1}{\delta}$. Strictly speaking, in a multivariate setting, one must be cautious about such substitutions when using big-$O$ notation. But we can always resolve potential issues by declaring that our complexity is always finite within a finite range in parameter space. 

Additionally, since $n$ is an integer, after the substitution we have ignored the regime of $n \ll \log \frac{1}{\delta}$ before the substitution. This makes our claim slightly weaker. Nevertheless, this format allows us to present the algorithms more succinctly. The omitted edge cases are of marginal importance and can be handled separately if desired.
} 
\end{theorem}

\begin{corollary}
\label{cry: classical_multi_estimator_relaxed}
For the purpose of this paper, we may relax Eq.~(\ref{eq: classical_multi_estimator}) in Theorem~\ref{thm: classical_multi_estimator} as:
\begin{equation}
\prob{\normtwo{\vec{\tilde{\mu}} - \expect{\vec{\RV{X}}}} \leqslant \sqrt{\frac{\tr \Sigma}{n}} \left(1 + \sqrt{\frac{1}{\ln \frac{1}{\delta}}}\right)} \geqslant 1 - \delta 
\end{equation}
\end{corollary}

We will also need to invoke the quantile estimation algorithm from Ref.~\cite{Hamoudi_2021}. The original theorem is:
\begin{theorem}[Qunantile Estimation]
\label{thm: quantile}
Given access to quantum experiment for multivariate random varaible $\RV{X}$, $p, \delta \in (0,1)$, there is a quantum algorithm that uses $O\left(\frac{\log \frac{1}{\delta}}{\sqrt{p}}\right)$ quantum experiments to return an approximate quantile $\tilde{Q}$ such that:
\begin{equation}
    Q(p) \leqslant \tilde{Q} \leqslant Q(Cp)
\end{equation}
with probability at least $1 - \delta$. $C \in (0,1)$ is some universal constant that is easy to compute. $Q$ denotes the quantile function defined as:
\begin{equation}
    Q(x) = \sup\{y \in \mathbb{R}: \prob{\RV{X} \geqslant y} \geqslant p\}
\end{equation}
\end{theorem}

\section{Key Subroutine Analysis}
\label{sec: key_subroutine}

With primitives all set, we now establish theorems for the key quantum subroutine used for multivariate mean value estimation. Intuitively, the idea is to build from the oracle in Definition~\ref{def: univariate phase oracle}, combined with the phase estimation procedure used in Ref.~\cite{Cornelissen_2022}. Eventually, we build several programs that give us insights into the expectation value of the random variable with high probability.

\subsection{Spectrum Analysis of The Grover Operator}

\label{sec: sectrum_Grover}

As we will see later, a successful univariate estimator is built upon an estimator for the univariate case. First, following the insight made by Ref.~\cite{Kothari_2022}, we construct the Grover operator with access to the phase oracle of a univariate random variable:

\begin{definition}[Univariate Grover Operator]
Given the phase oracle $\mathcal{O}$ for the univariate random variable $\RV{\theta}$, the Grover operator for $\RV{\theta}$ is 
\begin{equation}
\mathcal{G} = \mathcal{R} \mathcal{O}
\end{equation}
where $\mathcal{R}$ is the reflection gate with 
\begin{equation}
\mathcal{R} = \mathcal{P} (2 \ket{0}\bra{0} - I)\mathcal{P}^\dag = 2\ket{\bm{1}}\bra{\bm{1}} - I
\end{equation}
\end{definition}

In Sec.~3.7 of Ref.~\cite{Kothari_2022} the authors provide a geometric understanding of the eigenvectors and eigenvalues. Here, we would like to provide a more explicit form.

Fix a univariate random variable $\RV{\theta}$. Consider an eigenvector of the Grover operator with eigenvalue $e^{i\alpha}$. Say it is an associated quantum state of the complex random variable $\RV{\psi}$. Applying Grover operator gives
\begin{equation}
\begin{aligned}
  \mathcal{G}\ket{\RV{\psi}} &  = \left( {2\ket{\bm{1}}\bra{\bm{1}} - I} \right)U\ket{\RV{\psi}} = 2\ket{\bm{1}}\bra{\bm{1}}\ket{{e^{i\RV{\theta} }}\RV{\psi} } - \ket{{e^{i\RV{\theta} }}\RV{\theta} } \\ 
   &  = \ket{2\mathbb{E}\left[ e^{i\RV{\theta} }\RV{\psi}  \right] - {e^{i\RV{\theta} }}\RV{\psi} } \\ 
\end{aligned} 
\end{equation}

The eigenvalue equation is thus
\begin{equation}
    \ket{2\mathbb{E}\left[ e^{i\RV{\theta} }\RV{\psi}  \right] - {e^{i\RV{\theta} }}\RV{\psi} } = e^{i\alpha} \ket{\RV{\psi}} = \ket{e^{i\alpha} \RV{\psi}}
\end{equation}
In terms of the random variables:
\begin{equation}
    \label{eq: eigenvalue_eq}
    2\mathbb{E}\left[ e^{i\RV{\theta} }\RV{\psi}  \right] - {e^{i\RV{\theta}}}\RV{\psi} = e^{i\alpha} \RV{\psi}
\end{equation}

Despite the look, this equation is solvable. We find that ignoring edge cases, $\alpha$ needs to satisfy
\begin{equation}
\label{eq: alpha_eigen}
\expect{\tanft{\RV{\theta} - \alpha}} = 0
\end{equation}
and then the eigenvector, in random variable form, is given by
\begin{equation}
\label{eq: eigen_vecs}
\RV{\psi} = \frac{1 - i \tanft{\RV{\theta} - \alpha}}{\sqrt{1 + \expect{\tantft{\RV{\theta} - \alpha}}}} 
\end{equation}
The details of calculation is found in the following theorem. 

\begin{theorem}[Spectrum of Grover Operator]
\label{thm: spectrum_G}
The spectrum of the Grover operator $\mathcal{G}$ for a real univariate variable $\RV{\theta}$ contains all eigenvalues of the form $e^{i \alpha}$ where $-\pi < \alpha \leqslant \pi$ satisfies Eq.~(\ref{eq: alpha_eigen}) (assuming $\tanft{\theta_k - \alpha}$ do not blow up for any $k \in \Omega$). Corresponding to $\alpha$, the eigenvector, specified by $\ket{\RV{\psi}}$ where $\RV{\psi}$ is a complex random variable, satisfy Eq.~(\ref{eq: eigen_vecs}).

In the event that there are multiple outcomes in $\RV{\theta}$ that share the same value (mod $2\pi$). For each possible such value $\varphi \in (-\pi, \pi]$, let $S = \{k \in \Omega: \mathcal{\theta}_k = \varphi \pmod{2\pi}\}$. We find, in addition, an eigenvalue $e^{i\alpha}$ satisfying $\varphi - \alpha = \pi \pmod{2\pi}$, and eigenvectors $\ket{\RV{\psi}}$ (where $\RV{\psi}$ is a complex random variable) with 
\begin{equation}
\left\{
\begin{gathered}
    \psi_k = 0 \quad \forall k \not\in S \hfill \\
    \expect{\RV{\psi}} = 0 \hfill \\
\end{gathered}
\right.
\end{equation}
\end{theorem}

\begin{proof}
First, rewrite Eq.~(\ref{eq: eigenvalue_eq}) as:
\begin{equation}
\expect{e^{i (\RV{\theta} - \alpha)}\RV{\psi}} = \frac{1 + e ^ {i (\RV{\theta} - \alpha)}}{2} \RV{\psi}
\end{equation}

As a sanity check, $\alpha$ is covariant under translations in $\RV{\theta}$ (i.e., if we send $\RV{\theta} \to \RV{\theta} + \beta$ then $\alpha \to \alpha + \beta$). This equation implies that $\frac{ 1 + e ^ {i (\RV{\theta} - \alpha)} }{2} \RV{\psi}$ is a constant, i.e., 
\begin{equation}
\label{eq: simple_eigen_eq_1}
    \frac{1 + e ^ {i (\theta_k - \alpha)} }{2} \psi_k = C \quad \forall k \in \Omega
\end{equation}
where $C$ is some constant. But then 
\begin{equation}
    \begin{aligned}
        C & = \expect{e^{i (\RV{\theta} - \alpha)}\RV{\psi}} = \expect { \left(1 + e^{i (\RV{\theta} - \alpha)}\right)\RV{\psi}} - \expect{\RV{\psi}} \\
        & = 2C - \expect{\RV{\psi}} \\
    \end{aligned}
\end{equation}
so $C = \expect{\psi}$. The equation becomes:
\begin{equation}
\label{eq: simple_eigen_eq_2}
\expect{\RV{\psi}} = \frac{1 + e ^ {i (\RV{\theta} - \alpha)}}{2} \RV{\psi}
\end{equation}
By the same strategy one can also show that the above equation implies backwards to Eq.~(\ref{eq: simple_eigen_eq_1}). So they are equivalent. \footnote{We can also take the expectation on top of Eq.~(\ref{eq: simple_eigen_eq_1}) and also get $\expect{e^{i(\RV{\theta}-\alpha)}\RV{\psi}} = \expect{\RV{\psi}}$ and vise versa}

We now assume $\alpha$ is such that there is no $k \in \Omega$ such that $e^{i (\theta_k - \alpha)} = -1$. When this is not true (i.e., $\exists k \in \Omega \; e^{i (\theta_k - \alpha)} = -1$), the case can be found in Appendix~\ref{app: eigen_continued} because it is not relevant to later discussions in this paper.

With this assumption we derive
\begin{equation}
    \RV{\psi} = \frac{\expect{\RV{\psi}}}{\frac{1 + e^{i (\RV{\theta} - \alpha)}}{2}} = \expect{\RV{\psi}} \left( 1 - i \tanft{\RV{\theta} - \alpha}\right)
\end{equation}
where we used the fact that
\begin{equation}
\frac{1}{\frac{1+e^{i (\RV{\theta} - \alpha)}}{2}} = \frac{e^{-i \frac{\RV{\theta} - \alpha}{2}}}{ \frac{e^{-i \frac{\RV{\theta} - \alpha}{2}} + e^{i \frac{\RV{\theta} - \alpha}{2}}}{2}}  = \frac{e^{-i \frac{\RV{\theta} - \alpha}{2}}}{\cos \frac{\RV{\theta} - \alpha}{2} } = 1 - i \tanft{\RV{\theta} - \alpha} 
\end{equation}

We have now found the eigenvector corresponding to eigenvalue $e^{i \alpha}$. Note that $\expect{\RV{\psi}}$ is effectively a overall constant, which can be set such that the state $\ket{\RV{\psi}}$ is normalized. We find 
\begin{equation}
\bra{\RV{\psi}}\ket{\RV{\psi}} = \expect{\RV{\psi}}^2 \left(1 + \expect{\tantft{\RV{\theta} - \alpha}}\right)
\end{equation}

Ignoring phases we can pick $\expect{\RV{\psi}} = \frac{1}{\sqrt{\expect{\RV{\psi}}}}$, which gives 
\begin{equation}
    \RV{\psi} = \frac{1 - i \tanft{\RV{\theta} - \alpha}}{\sqrt{1 + \expect{\tantft{\RV{\theta} - \alpha}}}} 
\end{equation}

This proves the first half of the theorem. As mentioned, proof to the second half is not important but can be found in Appendix~\ref{app: eigen_continued}.
\end{proof}

A natural result is that if we were to do phase estimation with input state $\ket{\bm{1}}$, the probability of outcome $\alpha$ is $\frac{1}{1 + \expect{\tantft{\RV{\theta} - \alpha}}}$, as formally stated in the following remark.
\begin{corollary}[Probability]
\label{cry: phase_probability}
Given univariate random variable $\RV{X}$ and $\RV{\theta} = 2 \arctan \RV{X}$. Let $-\pi < \alpha \leqslant \pi$ satisfy $\expect{\tanft{\RV{\theta} - \alpha}} = 0$ (and that $\tanft{\RV{\theta}- \alpha}$ does not blow up). 

The Grover operator $\mathcal{G}$ from $\RV{X}$ satisfies:
\begin{equation}
    \label{eq: prob_Grover}
    \probM{\RV{\varphi}}{\arg \mathcal{G}}{\ket{\bm{1}}}{\RV{\varphi} = \alpha} = \left|\bra{\RV{\psi}}\ket{\bm{1}}\right|^2  = \frac{1}{1 + \expect{\tantft{\RV{\theta} - \alpha}}}
\end{equation}
where $\ket{\RV{\psi}}$ is the eigenvector just specified in Theorem~\ref{thm: spectrum_G}
\end{corollary}

Combining results from Theorem~\ref{thm: spectrum_G} and Corollary~\ref{cry: phase_probability}, we can approximately understand the spectrum of the Grover operator. As an intuitive justification of our later results, we first rewrite Eq.~(\ref{eq: alpha_eigen}) as:
\begin{equation}
    \label{eq: alpha_eigen_n1}
    \expect{\frac{\tanft{\RV{\theta}} - \tanft{\alpha}}{1 + \tanft{\RV{\theta}}\tanft{\alpha}}} = 0
\end{equation}
If we made assumptions such that $\tanft{\RV{\theta}}\tanft{\alpha}$ is sufficiently small, then the above equation simply reduces to 
\begin{equation}
    \tanft{\alpha} \approx \expect{\tanft{\RV{\theta}}}
\end{equation}
This is a great result. However, in practice $\tanft{\theta_k}$ can be very large for some $k \in \Omega$. During mean value estimation, the algorithm should iteratively refine the estimation of the mean value, so the best assumptions we can make is that $\expect{\tanft{\RV{\theta}}} = O(\varepsilon)$ and $\tanft{\theta_k} = O(\frac{1}{\varepsilon})$. In that case we expect $\tanft{\alpha} = O(\varepsilon)$ such that $\tanft{\RV{\theta}}\tanft{\alpha}$ is some sufficiently small constant. Additionally, $\expect{\tantft{\RV{\theta}}} = s^2$ is some small constant.  

Using the fact that $\frac{1}{1-x} = \sum_{n \geqslant 0} x^n$ we can rewrite Eq.~(\ref{eq: alpha_eigen_n1}) as:
\begin{equation}
    \begin{aligned}
          \expect{\frac{\tanft{\RV{\theta}} - \tanft{\alpha}}{1 + \tanft{\RV{\theta}}\tanft{\alpha}}}  & = \sum_{n \geqslant 0} \expect{(\tanft{\RV{\theta}} - \tanft{\alpha})\left(- \tanft{\RV{\theta}} \tanft{\alpha}\right)^n} \\
        & = \sum_{n \geqslant 0} (-1)^n \left(\expect{\tannft[n+1]{\RV{\theta}}} \tannft[n]{\alpha} + \expect{\tannft[n]{\RV{\theta}}} \tannft[n+1]{\alpha}\right) \\
        & = 0
    \end{aligned}
\end{equation}

Our previous assumption means that $\expect{\tannft[n+2]{\RV{\theta}}} \leqslant \expect{|\tanft{\RV{\theta}}|^{n+2}}= s^2 O(\frac{1}{\varepsilon})^n$ for $n \geqslant 0$, where $O(\frac{1}{\varepsilon})$ represents some small constant divided by $\varepsilon$. (so the ``1'' in big $O$ specifically refers to a small constant) The above equation then simplifies to:
\begin{equation}
    \begin{aligned}
          \expect{\frac{\tanft{\RV{\theta}} - \tanft{\alpha}}{1 + \tanft{\RV{\theta}}\tanft{\alpha}}} & = \expect{\tanft{\RV{\theta}}} - s^2 \sum_{n\geqslant 0} (-1)^n O\left(\frac{1}{\varepsilon}\right)^{n} \tannft[n+1]{\alpha}  - \tanft{\alpha} + O(\varepsilon^3) \\
        & = \expect{\tanft{\RV{\theta}}} - \left(1 + s^2 \sum_{n\geqslant 0}(-1)^n O(1)^{n} \right) \tanft{\alpha} + O(\varepsilon^3) \\
        & = \expect{\tanft{\RV{\theta}}} - \left(1 + s^2 O(1)\right) \tanft{\alpha} + O(\varepsilon^3) \\
        & = 0
    \end{aligned}
\end{equation}
where $O(1)$ represents constants small enough such that the summation $\sum_{n\geqslant 0}(-1)^n O(1)^{n}$ do not diverge. 

This implies that 
\begin{equation}
    \tanft{\alpha} = \frac{\expect{\tanft{\RV{\theta}}}}{1 + s^2 O(1)} + O(\varepsilon^3)
\end{equation}

Since $\alpha$ is itself a small quantity, we might as well write is as:
\begin{equation}
    \begin{aligned}
        \alpha & = 2 \frac{\expect{\tanft{\RV{\theta}}}}{1 + s^2 O(1)} + O(\varepsilon^3) \\
        & = 2 \expect{\tanft{\RV{\theta}}}\left(1 - s^2 O(1)\right) \\
    \end{aligned}
\end{equation}
where we used the fact that $\varepsilon \leqslant s$, as $\expect{\tanft{\RV{X}}}^2 \leqslant \expect{\tantft{\RV{X}}}$.

Since $\expect{\tantft{\theta}} = s^2$, the intuitive understanding is that we set the constants in our assumptions such that for each $k \in \Omega$, $\tanft{\theta_k + \alpha}$ do not differ from $\tanft{\theta_k}$ by more than $O(1)$ factor. Therefore we find that $\expect{\tantft{\theta + \alpha}} = s^2 O(1)$, which gives the measurement probability
\begin{equation}
    \probM{\RV{\varphi}}{\arg \mathcal{G}}{\ket{\bm{1}}}{\RV{\varphi} = \alpha} = 1 - s^2 O(1)
\end{equation}

The above discussions offer an intuitive understanding of eigenvalue of the Grover operator near $\alpha = 2 \expect{\tanft{\theta}}$. We formalize the statements in Theorem~\ref{thm: alpha_close} below. 
\begin{theorem}[Key Property of Grover Operator Spectrum]

\label{thm: alpha_close}

Fix some parameters $\varepsilon, \lambda, s_0$ with $\varepsilon \leqslant s_0 \leqslant \frac{1}{6}$, $\lambda \geqslant 5$, and $c \leqslant \frac{1}{4}$.

Consider random variable $\RV{\theta}$ satisfying 
\begin{itemize}
    \item $\forall k \in \Omega \; |\tanft{\theta_k}| \leqslant \frac{1}{\lambda \varepsilon}$;
    \item $\left|\expect{\tanft{\RV{\theta}}}\right| \leqslant \varepsilon$;
    \item $\expect{\tantft{\RV{\theta}}} \leqslant s_0^2$;
\end{itemize}
Let $\mathcal{G}$ be the Grover operator for $\RV{\theta}$. There is an eigenvalue $e^{i\alpha}$ (where $\alpha \in (-\pi, \pi]$) such that 
\begin{itemize}
    \item $\left|\alpha - 2\expect{\tanft{\RV{\theta}}}\right| \leqslant 2 c \varepsilon$;
    \item $\probM{\RV{\varphi}}{\arg \mathcal{G}}{\ket{\bm{1}}}{\RV{\varphi} = \alpha} \geqslant 1 - \delta$;
\end{itemize}
with the added constraint:
\begin{itemize}
    \item $c \geqslant \frac{7.635 s_0^2}{1+s_0^2}$; \footnote{It is clear $\frac{7.635 s_0^2}{1+s_0^2} \leqslant \frac{1}{4}$ so $c$ exists}.
    \item $\delta \geqslant 1.7983 s_0^2 + 7.480 s_0 \varepsilon$
\end{itemize}
\end{theorem}

\begin{proof}

Without loss of generality set $ \forall k\in\Omega \; \theta_k, \alpha \in (-\pi,\pi]$. Let $\expect{\tanft{\RV{\theta}}} = \mu$, $\Var \left(\tanft{\RV{\theta}}\right) = \sigma^2$, and $\expect{\tantft{\RV{\theta}}} = s^2 = \sigma^2 + \mu^2$.  

Define a function 
\begin{equation}
    f(\beta) = \expect{\tanft{\RV{\theta} - \beta}}
\end{equation}
where $\beta$ is restricted to range
\begin{equation}
    B = [-2 (1+c) \varepsilon, 2 (1+c) \varepsilon)]
\end{equation}
First, we would like to show that there is some $\alpha$ from the range:
\begin{equation}
    A = \left[ 2\left(\mu - c \varepsilon\right), 2\left(\mu + c \varepsilon\right)\right]
\end{equation}
such that such that $f(\alpha)$ intersects the x axis. Such $\alpha$ will imply an eigenvalue of $e^{i\alpha}$ for $\mathcal{G}$ via Theorem~\ref{thm: spectrum_G}, as we can verify the the endpoints of $B$ are within $(-\pi, \pi]$.  It is clear that $A \subseteq B$. 

First, we want to show that the function is continuous in range $B$. To do that, first, we bound $\tanft{\beta}$ with Taylor remainder theorem:
\begin{equation}
    | \tanft{\beta} | \leqslant \frac{|\beta|}{2} + \frac{C |\beta|^3}{6} 
    \leqslant (1 + c + \kappa) \varepsilon 
\end{equation}
where $C = \left. \pdv[3]{\beta} ( \tanft{\beta})\right|_{\beta = (1+c) \varepsilon} \leqslant 0.2962 $ (because we can show that $\pdv[3]{\beta} ( \tanft{\beta})$ is an even function that is increasing function in $\beta \in B \cup \mathbb{R}^+$), and $\kappa = \frac{4}{3} C (1+c)^3 \varepsilon^2 \leqslant 0.3950  (1+c)^3 \varepsilon^2$ or simply $\kappa \leqslant 0.02143$. We use the fact that $(1+c) \varepsilon \leqslant \frac{5}{12}$. We can further define $\eta = 1 + c + \kappa \leqslant 1.27143$ such that $\tanft{\beta} \leqslant \eta \varepsilon$.

Fix $\beta \in I$. It is clear that  $\lambda > \eta$. Thus $\forall k \in \Omega \;  |\tanft{\theta_k} \tanft{\beta}| \leqslant \frac{\eta}{\lambda} < 1$, i.e., $|\tanft{\theta_k}| \leqslant \left|\frac{1}{\tanft{\beta}}\right|$. This means that $|\frac{\beta}{2}| \leqslant \frac{\pi}{2} - |\frac{\theta_k}{2}|$. Therefore, $|\frac{\theta_k - \beta}{2}| \leqslant \frac{\pi}{2}$. This shows $f$ is continuous in $B$ and hence also non-increasing because $\tan$ is a non-decreasing function. So, to show that a solution exists in $A$, it suffices to consider the endpoints of interval $A = [2(\mu - c \varepsilon), 2(\mu + c \varepsilon)]$ and show that $f(2(\mu - c \varepsilon)) \geqslant 0$ and $f(2(\mu + c \varepsilon)) \leqslant 0$.

For later convenience, for $k \in \Omega$ we continue to find
\begin{equation}
    \label{eq: lemma_shift_factor}
    \left|\tanft{\theta_k - \beta}\right| \leqslant \frac{\left|\tanft{\theta_k}\right| + \left|\tanft{\beta}\right|}{1 - \left|\tanft{\theta_k} \tanft{\beta}\right|} \leqslant \left(\left|\tanft{\theta_k}\right| + \eta \varepsilon \right)\frac{\lambda}{\lambda-\eta} = \chi \left( \left|\tanft{\theta_k}\right| + \eta\varepsilon \right)
\end{equation}
For simplicity we define $\chi = \frac{\lambda}{\lambda - \eta}$. By Cauchy-Schwarz inequality we obtain
\begin{equation}
    \label{eq: 1-norm_bound}
    \expect{\left|\tanft{\RV{\theta}}\right|} \leqslant \sqrt{\expect{\tannft[2]{\RV{\theta}}}\expect{\bm{1}^2}} = s
\end{equation}
With Eq.~(\ref{eq: lemma_shift_factor}) and Eq.~(\ref{eq: 1-norm_bound}) we find:
\begin{equation}
    \label{eq: 1-norm_bound_shift}
    \expect{\left|\tanft{\RV{\theta} - \beta}\right|} \leqslant \chi \left( s + \eta \varepsilon \right)
\end{equation}
along with 
\begin{equation}
    \label{eq: 2-norm_bound_shift}
    \begin{aligned}
        \expect{\tantft{\RV{\theta} - \beta}} 
        & \leqslant \chi^2 \left( \expect{\tantft{\RV{\theta}}} + 2 \eta  \expect{\left|\tanft{\RV{\theta}}\right|} + \eta^2 \varepsilon^2 \right) \\
        & \leqslant \chi^2(s+\eta\varepsilon)^2 \leqslant \chi^2  (s^2 + (\eta^2 + 2 \eta) s \varepsilon)
    \end{aligned}
\end{equation}
Meanwhile, from $\left|\tanft{\theta_k}\right| \leqslant \frac{1}{\lambda\varepsilon} \; \forall k \in \Omega$ we obtain $\expect{\tannft[3]{\RV{\theta}}} \leqslant \frac{1}{\lambda\varepsilon} s^2$. With the same strategy we can then show
\begin{equation}
    \label{eq: 3-norm_bound_shift}
    \begin{aligned}
        \expect{\left|\tannft[3]{\RV{\theta} - \beta}\right|} & \leqslant \chi^3 \left(\frac{s^2}{\lambda \varepsilon} + 3 s^2 \eta \varepsilon + 3 s \eta^2 \varepsilon^2 + \eta ^3 \varepsilon^3 \right) \\
        & \leqslant \chi^3 \left(\frac{s^2}{\lambda \varepsilon} + (\eta^3 + 3 \eta^2 + 3 \eta) s^2 \varepsilon \right) \\ 
    \end{aligned}
\end{equation}

We note the derivatives of $f$:
\begin{equation}
    \left\{\begin{gathered}
        f'(\beta) = -\frac{1}{2} \left(1 + \expect{\tantft{\RV{\theta} - \beta}}\right)\hfill \\
        \hfill \\
        f''(\beta) = \frac{1}{2}\left(\expect{\tanft{\RV{\theta} - \beta}} + \expect{\tannft[3]{{\RV{\theta} - \beta}}}\right)
    \end{gathered}\right.
\end{equation}
This gives a Taylor expansion around $\beta = 0$:
\begin{equation}
f(\beta) = f(0) + f'(0) \beta + \frac{1}{2} f''(\beta') \beta^2 = \mu - \frac{1}{2} (1 + s^2) \beta + \frac{1}{2} f''(\beta') \beta^2
\end{equation}
where $\beta'$ is between $0$ and $\beta$. 

With results from Eqs.~(\ref{eq: 1-norm_bound_shift}) and (\ref{eq: 3-norm_bound_shift}) we bound $f''(\beta')$ with 

\begin{equation}
    \begin{aligned}
        \left|f''(\beta')\right| & \leqslant \chi\left(\frac{1}{2} s + \frac{\eta}{2} \varepsilon\right) + \chi^3 s^2 \left(\frac{1}{2 \lambda \varepsilon} + \frac{\eta^3 + 3\eta^2 + \eta}{2} \varepsilon\right) \\
        & \leqslant \chi^3 s^2 \frac{1}{2 \lambda \varepsilon} + s \left(\chi \frac{1+\eta}{2} + \chi^3 \frac{\eta^3 + 3 \eta^2 + 3 \eta}{2} s \varepsilon \right) s \\
        & = \chi^3 s^2 \frac{1}{2 \lambda \varepsilon} + D s  \\
    \end{aligned} 
\end{equation}
where we used the assumption $\varepsilon \leqslant s$. We also defined:
\begin{equation}
D = \chi \frac{1+\eta}{2} + \chi^3 \frac{\eta^3 + 3 \eta^2 + 3 \eta}{2} s \varepsilon \leqslant  1.13572 \chi + 0.14888 \chi^3
\end{equation}

Let's define $g(\beta) = \mu - \frac{1}{2} (1+s^2) \beta$ to be the linear approximation around for $f(\beta)$. We find that:
\begin{equation}
    \left\{ \begin{gathered}
        g\left(2(\mu - c \varepsilon)\right) = - s^2 \mu + (1+s^2) c \varepsilon \geqslant (c + s^2(c-1)) \varepsilon \hfill \\
    g\left(2(\mu + c \varepsilon)\right) = - s^2 \mu - (1+s^2) c \varepsilon \leqslant -(1+s^2) c \varepsilon \hfill \\
    \end{gathered} \right.
\end{equation}
Since $-(1+s^2)c = - (c+s^2c) \leqslant - (c+s^2c - c) = - (c+s^2(c-1)) $ we combine these two equations into:
\begin{equation}
    \label{eq: bound_g}
    \left\{ \begin{gathered}
        g\left(2(\mu - c \varepsilon)\right) \geqslant (c + s^2(c-1)) \varepsilon \hfill \\
        g\left(2(\mu + c \varepsilon)\right) \leqslant - (c + s^2(c-1)) \varepsilon \hfill \\
    \end{gathered} \right.
\end{equation}

Let the difference between $f$ and $g$ at $2(\mu \pm c \varepsilon)$ be $\Delta_{\pm}$, we can bound it with
\begin{equation}
    \label{eq: bound_Delta}
    \begin{aligned}
         \Delta_{\pm} & = \left| f \left(2 (\mu \pm c \varepsilon) \right)  - g\left(2 (\mu \pm c \varepsilon) \right)\right| = \frac{1}{2} |f'(\beta'_{\pm})| (2 (\mu \pm c\varepsilon))^2 \\
        & \leqslant \frac{1}{2} \left(  \chi^3 s^2 \frac{1}{2 \lambda \varepsilon} + D s \right) \times 4 (1+c)^2 \varepsilon^2 = \chi^3 (1+c)^2 \frac{1}{\lambda} s^2 \varepsilon + 2D(1+c)^2 s \varepsilon^2 \\
        & \leqslant (1+c)^2 \left( \frac{\chi^3}{\lambda} + 2D \right) s^2 \varepsilon
    \end{aligned}
\end{equation}
where $\beta'_{-}, \beta'_+ \in B$, $(\mu \pm c\varepsilon)^2 \leqslant (1+c)^2 \varepsilon^2$ because $|\mu| \leqslant \varepsilon$.

Using $\lambda \geqslant 5$, we obtain $\chi \leqslant 1.3410$ and $D \leqslant 1.88203$. This gives $\Delta_{\pm} \leqslant 6.635$ via Eq.~(\ref{eq: bound_Delta}). The constraint we defined in theorem statement $c \geqslant \frac{7.635 s_0^2}{1+s_0^2}$ thus becomes:
\begin{equation}
    \begin{gathered}
        c \geqslant \frac{7.635 s_0^2}{1+s_0^2} \geqslant \frac{7.635 s^2}{1+s^2} \hfill \\
        \therefore (1+s^2)c - s^2 \geqslant 6.635 s^2 \hfill \\ 
        \therefore (c+(s^2)(c-1)) \varepsilon \geqslant 6.635s^2 \varepsilon \geqslant \Delta_{\pm} \hfill \\ 
    \end{gathered} 
\end{equation}
Combined with Eq.~(\ref{eq: bound_g}) we then reach the conclusion that $f(2(\mu - c\varepsilon)) \geqslant 0$ and $f(2(\mu+c \varepsilon)) \leqslant 0$. This demonstrates that there is a solution $\alpha \in [2(\mu - c \varepsilon), 2 (\mu + c \varepsilon)]$. 

As the second part of the proof, we show that the probability for $\alpha$ is sufficiently high. First, Eq.~(\ref{eq: 2-norm_bound_shift}) becomes
\begin{equation}
    \expect{\tantft{\RV{\theta} - \beta}} \leqslant 1.7983 s^2 + 7.480 s \varepsilon 
\end{equation}
Clearly $\expect{\tantft{\RV{\theta} - \beta}} < 1$, by Eq.~(\ref{eq: prob_Grover}) we obtain:
\begin{equation}
    \begin{aligned}
        \probM{\RV{\varphi}}{\arg \mathcal{G}}{\ket{\bm{1}}}{\RV{\varphi} = \alpha} & = \frac{1}{1 + \expect{\tantft{\RV{\theta} - \alpha}}} \\
        & \geqslant 1 - \expect{\tantft{\RV{\theta} - \alpha}} \\
        & \geqslant  1 - (1.7983 s^2 + 7.480 s \varepsilon) \\
        & \geqslant 1 - \delta
    \end{aligned}
\end{equation}
The last inequality can be derived from the constraint $\delta \geqslant 1.7983 s_0^2 + 7.480 s_0 \varepsilon$.
\end{proof}



\subsection{Quantum Subroutine}   

Theorem~\ref{thm: alpha_close} gives a pretty good result of the properties of the Grover operator. However, there are some nasty constraints and irregularities, namely a factor 2 in the eigenvalue, a lot of $\tan$ functions, and an enforced upper bound for all instances of the random variables. The following theorem cleans things up:

\begin{theorem}[Theorem~\ref{thm: spectrum_G} Packaged]
    \label{thm: alpha_RV}
    Consider a univariate random variable $\RV{X}$ with
    \begin{itemize}
        \item $\left| \expect{\RV{X}} \right| \leqslant \varepsilon$;
        \item $\expect{\RV{X}^2} \leqslant s_0^2$
    \end{itemize}
    where $\varepsilon \leqslant s_0 \leqslant \frac{1}{3}$.
    Consider the Grover operator $\mathcal{G}$ on random variable $\RV{\theta} = 2 \arctan{\left(\frac{1}{2} \truncate{\RV{X}}{\frac{1}{\lambda \varepsilon}}\right)}$ with $\lambda = \frac{5}{4 - 5 s_0^2}$, it has an eigenvalue $e^{i\alpha}$ with eigenvector $\ket{\alpha}$, where $\alpha \in (-\pi, \pi]$, such that 
    \begin{itemize}
        \item $\left|\alpha - \expect{\RV{X}}\right| \leqslant \frac{3.1588 s_0^2}{1 - 1.25s_0^2}\varepsilon$ ;
        \item $\left|\bra{\alpha}\ket{\bm{1}}\right|^2 \geqslant 1 - \frac{1}{4} \left( 1.7983 s_0^2 + 7.480 s_0 \frac{\varepsilon}{1 - 1.25 s_0^2} \right)$.
    \end{itemize}
\end{theorem}

\begin{proof}

Note that $\tanft{\RV{\theta}} = \frac{1}{2} \truncate{\RV{X}}{\frac{1}{\lambda \varepsilon}}$. By Cauchy-Schwarz Inequality we find:
\begin{equation}
    \label{eq: alpha_RV_proof_1}
    \left| \expect{\RV{X}} - \expect{\truncate{\RV{X}}{\frac{1}{\lambda \varepsilon}}} \right| \leqslant \expect{\left|\RV{X} - \truncate{\RV{X}}{\frac{1}{\lambda \varepsilon}}\right|} \leqslant \expect{\left|\RV{X}\right|\bm{1}_{\left|\RV{X}\right| > \frac{1}{\lambda\varepsilon}}} \leqslant \sqrt{\expect{\RV{X}^2} \prob{\left|\RV{X}\right| > \frac{1}{\lambda \varepsilon}}}
\end{equation}
With Markov inequality applied to $\RV{X}^2$ we find:
\begin{equation}
\label{eq: alpha_RV_proof_2}
\prob{\left|\RV{X}\right| > \frac{1}{\lambda \varepsilon}} \leqslant \frac{\expect{\RV{X}^2}}{\frac{1}{\lambda\varepsilon}} \leqslant (\lambda \varepsilon)^2 s_0^2
\end{equation}
Combining Eqs.~(\ref{eq: alpha_RV_proof_1} and \ref{eq: alpha_RV_proof_2}) we obtain:
\begin{equation}
\label{eq: alpha_RV_deviate_truncate}
\left| \expect{\RV{X}} - \expect{\truncate{\RV{X}}{\frac{1}{\lambda \varepsilon}}} \right| \leqslant \lambda s_0^2 \varepsilon
\end{equation}
Thus, it is easy to see the following three results:
\begin{itemize}
    \item $\forall k\in\Omega \quad \left|\tanft{\theta_k}\right| \leqslant \frac{1}{\frac{4\lambda}{1+\lambda s_0^2} \frac{1}{2} (1+\lambda s_0^2) \varepsilon}$;
    \item $\expect{\tantft{\RV{\theta}}} \leqslant \frac{1}{4}s_0^2$;
    \item $\expect{\tanft{\RV{\theta}}} \leqslant \frac{1}{2}(1 + \lambda s_0^2) \varepsilon$;
\end{itemize}
By picking $\lambda = \frac{5}{4 - 5 s_0^2} \in [1.25, 1.452)$ we make $\frac{4\lambda}{1+\lambda s_0^2} = 5$. Then, Theorem~\ref{thm: alpha_close} holds with parameters $(\varepsilon, s_0, \lambda)$ replaced with $(\min\left\{\frac{1}{2}(1 + \lambda s_0^2) \varepsilon, \frac{1}{2}s_0\right\}, \frac{1}{2} s_0, 5)$. The min function is placed in the event that $(1 + \lambda s_0^2) \varepsilon > s_0$, in which case we still have $\expect{\tanft{\RV{\theta}}} \leqslant \frac{1}{2} s_0$ by Cauchy-Schwarz inequality, and also  $\forall k\in\Omega \quad \left|\tanft{\theta_k}\right| \leqslant \frac{1}{5 \frac{1}{2} (1+\lambda s_0^2) \varepsilon} \leqslant \frac{1}{5s_0}$. 

Note that $\frac{1}{2}(1 + \lambda s_0^2) \varepsilon = \frac{1}{2} \frac{\varepsilon}{1-1.25 s_0^2}$. Theorem~\ref{thm: alpha_close} says there is a solution $\alpha$ such that  
\begin{equation}
     \left|\alpha - 2 \expect{\tanft{\RV{\theta}}}\right| = \left|\alpha - \expect{\truncate{\RV{X}}{\frac{1}{\lambda \varepsilon}}}\right| \leqslant 2 c \left(\frac{1}{2}(1 + \lambda s_0^2)\varepsilon\right) = \frac{c \varepsilon}{1-1.25 s_0^2}
\end{equation}
where we take $c = 1.9088 s_0^2 \geqslant \frac{7.635 (\frac{1}{4} s_0^2)} {1 + \left(\frac{1}{4} s_0^2\right)}$. Combined with Eq.~(\ref{eq: alpha_RV_deviate_truncate}), we find 
\begin{equation}
    \left|\alpha - \expect{\RV{X}}\right| \leqslant \left(\frac{c}{1-1.25 s_0^2} + \lambda s_0^2\right) \varepsilon
    \leqslant \frac{3.1588 \varepsilon s_0^2}{1 - 1.25s_0^2}
\end{equation}

Moreover, let the corresponding eigenstate (with eigenvalue $e^{i \alpha}$ for $\mathcal{G}$ be $\ket{\alpha}$ with norm 1, the theorem also states that $\left|\bra{\alpha}\ket{\bm{1}}\right|^2 \geqslant 1 - \delta$, where $\delta = \frac{1}{4} \left( 1.7983 s_0^2 + 7.480 s_0 \frac{\varepsilon}{1 - 1.25 s_0^2} \right)$. We can verify that the above derivation process is still valid in the event that $(1+ \lambda s_0^2) \varepsilon > s_0$. 
\end{proof}

\begin{remark}[Constructing Grover Operator]
\label{remark: alpha_RV}
Given $O(1)$ access to a quantum experiment (and its inverse) to the random variable $\RV{X}$, by Remark~\ref{remark: quantum_experiment_algebra} we can construct a quantum experiment for $\RV{\theta}$. By Definition~\ref{def: univariate phase oracle} we can construct the phase oracle for $\RV{\theta}$ which gives us access to $\mathcal{G}$ discussed in Theorem~\ref{thm: alpha_RV}
\end{remark}

Continue on from Theorem~\ref{thm: alpha_RV}, we can then give a bound in terms of distance between quantum states.
\begin{theorem}[State Distance Bound for Grover Operator]
    \label{thm: alpha_RV_diff}
    Consider a univariate random variable $\RV{X}$ with
    \begin{itemize}
        \item $\left| \expect{\RV{X}} \right| \leqslant \varepsilon$;
        \item $\expect{\RV{X}^2} \leqslant s_0^2$
    \end{itemize}
    where $\varepsilon \leqslant s_0 \leqslant \frac{1}{3}$.
    Consider the Grover operator $\mathcal{G}$ on random variable $\RV{\theta} = 2 \arctan{\left(\frac{1}{2} \truncate{\RV{X}}{\frac{1}{\lambda \varepsilon}}\right)}$ with $\lambda = \frac{5}{4 - 5 s_0^2}$, we claim that for $N \in \mathbb{N}$
    \begin{equation}
        \norm{\mathcal{G}^N \ket{\bm{1}} -  e^{i N \expect{\RV{X}}} \ket{\bm{1}}}^2 
        \leqslant 1.7983 s_0^2 + \frac{7.480}{1 - 1.25s_0^2} s_0 \varepsilon + \left(\frac{3.1588 N \varepsilon s_0^2}{1 - 1.25 s_0^2}\right)^2 
    \end{equation}
\end{theorem}
\begin{proof}

Theorem~\ref{thm: alpha_RV} tells us that $\mathcal{G}$ has eigenvalue $e^{i\alpha}$ with eigenket $\ket{\alpha}$, where $\alpha \in (-\pi, \pi]$, such that 
\begin{itemize}
        \item $\left|\alpha - \expect{\RV{X}}\right| \leqslant \frac{3.1588 s_0^2}{1 - 1.25s_0^2} \varepsilon$ ;
        \item $\left|\bra{\alpha}\ket{\bm{1}}\right|^2 \geqslant 1 - \frac{1}{4} \left( 1.7983 s_0^2 + 7.480 s_0 \frac{\varepsilon}{1 - 1.25 s_0^2} \right)$.
\end{itemize}

Consider the projection of $\ket{\bm{1}}$ onto $\ket{\alpha}$, $\bra{\alpha}\ket{\bm{1}}\ket{\alpha}$. By triangle inequality we know that 
\begin{equation}
    \mathcal{G}^N \ket{\bm{1}} -  e^{i N \expect{\RV{X}}} \ket{\bm{1}} = \ket{\Delta_1} + \ket{\Delta_2} + \ket{\Delta_3}
\end{equation}
where
\begin{equation}
    \left\{\begin{gathered}
        \ket{\Delta_1}  = \mathcal{G}^N \ket{\bm{1}} - \mathcal{G}^N  \bra{\alpha}\ket{\bm{1}}\ket{\alpha}  =  \mathcal{G}^N \ket{\bm{1}} - e^{i N \alpha}  \bra{\alpha}\ket{\bm{1}}\ket{\alpha} \\  \hfill \\
        \ket{\Delta_2} =  e^{i N \alpha} \bra{\alpha}\ket{\bm{1}}\ket{\alpha}  - e^{i N \expect{\RV{X}}}  \bra{\alpha}\ket{\bm{1}}\ket{\alpha} \hfill \\
        \ket{\Delta_3} = e^{i N \expect{\RV{X}}}  \bra{\alpha}\ket{\bm{1}}\ket{\alpha} - e^{i N \expect{\RV{X}}} \ket{\bm{1}} \hfill \\
    \end{gathered} \right.
\end{equation}
Or, to simplify a bit:
\begin{equation}
    \left\{\begin{gathered}
        \ket{\Delta_1} = \mathcal{G}^N \left(\ket{\bm{1}} -  \bra{\alpha}\ket{\bm{1}}\ket{\alpha} \right) \hfill \\
        \ket{\Delta_2} =  \left( e^{i N \alpha} - e^{i N \expect{\RV{X}}} \right) \bra{\alpha}\ket{\bm{1}}\ket{\alpha} \hfill \\
        \ket{\Delta_3} = e^{i N \expect{\RV{X}}}  \left( \bra{\alpha}\ket{\bm{1}}\ket{\alpha} - \ket{\bm{1}} \right) \hfill \\
    \end{gathered} \right.
\end{equation}
It is thus easy to see that $\ket{\Delta_1} \perp \ket{\Delta_2}$ and $ \ket{\Delta_3} \perp \ket{\Delta_2}$. Combined we can use triangular inequality to bound:
\begin{equation}
    \label{eq: alpha_RV_proof_3}
    \begin{aligned}
        \left\lVert\mathcal{G}^N \ket{\bm{1}} -  e^{i N \expect{\RV{X}}} \ket{\bm{1}}\right\rVert^2 & \leqslant \biggl| \bigl\lVert \ket{\Delta_1} \bigr\rVert + \bigl\lVert \ket{\Delta_3} \bigr\rVert \biggr|^2 + \bigl\lVert\ket{\Delta_2} \bigr\rVert ^2 \\
        & \leqslant 4 \delta + \left(\frac{3.1588 N \varepsilon s_0^2}{1 - 1.25s_0^2}\right)^2 \\
        & \leqslant 1.7983 s_0^2 + \frac{7.480}{1 - 1.25s_0^2} s_0 \varepsilon + \left(\frac{3.1588 N \varepsilon s_0^2}{1 - 1.25s_0^2}\right)^2 \\
    \end{aligned}
\end{equation}
where we used the fact that $\bigl\lVert \ket{\Delta_1} \bigr\rVert =  \bigl\lVert \ket{\Delta_3} \bigr\rVert \leqslant \delta$ and 
\begin{equation}
\bigl\lVert \ket{\Delta_2} \bigr\rVert = \left|e^{i N \alpha} - e^{i N \expect{\RV{X}}}\right| \left|\bra{\alpha}\ket{\bm{1}}\right| \leqslant N |\alpha - \expect{\RV{X}}|\sqrt{1-\delta} \leqslant  \frac{3.1588 N \varepsilon s_0^2}{1 - 1.25 s_0^2}
\end{equation}

Eq.~(\ref{eq: alpha_RV_proof_3}) is exactly our conclusion. 

\end{proof}

Theorem~\ref{thm: alpha_RV_diff} and Remark~\ref{thm: alpha_RV_diff} presented a primitive that serves the same purpose as the ``Directional Mean Oracle'' in Ref.~\cite{Cornelissen_2022}, except there are no truncations and the fact that the deviation from $e^{i N \expect{\RV{X}}}$ cannot be made arbitarily small. This would eventually results in a different multivariate mean value estimator in Sec.~\ref{sec: simple_algo}.

\subsection{Univariate Mean Value Estimator}
\label{sec: uni_estimator}
Theorem~\ref{thm: alpha_RV} has been begging for us to build an intuitive mean value estimator, in which we exponentially decrease our confidence parameter $\varepsilon$ to give a more accurate estimation of the random variable. First, knowing that the mean is small, the following algorithm gives a closer mean estimate. This can be thought as a ``refinement'' step because say we are trying to estimate some $\RV{X}$, knowing a mean estimate $\tilde{\mu}$, running our program on $\RV{X} - \tilde{\mu}$ leads to a more accurate estimate of the mean. 
\begin{algorithm}
    \caption{\label{alg: uni_refinement} Refinement Step of Univariate Mean Value Estimator}
    \KwData{Access to Quantum Experiment of Random Variable $\RV{X}$, accuracy parameter $\varepsilon \geqslant 0$, confidence parameter $0 < \delta < 1$}
    \KwResult{An mean value estimate $\tilde{\mu}$}
    $M \leftarrow  2\left\lceil\frac{18 \ln \frac{1}{\delta} - 1}{2}\right\rceil + 1$\;
    $N \leftarrow 2^{\left\lceil\log_2 \frac{24\pi}{\varepsilon}\right\rceil}$\;
    Initialize space for vector arrays $u$, $\mu$ of dimension $d$ and length $M$ \;
    \For{$\ell = 1 $ \KwTo $M$}{
        $u_\ell \leftarrow $ result of phase estimation algorithm on Grover operator $\mathcal{G}$ with $N$ access to phase oracle of random variable $\RV{\theta} = 2 \arctan{\left(\frac{1}{2} \truncate{\RV{X}}{\frac{1}{\lambda \varepsilon}}\right)}$ where $\lambda = \frac{5}{4 - 5 \left(\frac{\sqrt{10}}{12}\right)^2}$ \;
        $\mu_\ell \leftarrow 2 \pi u_\ell$ \;
    }
    $\tilde{\mu} \leftarrow \text{ median of } \{\ell \in [M]: \mu_\ell\}$\;
    Output $\tilde{\mu}$\;
\end{algorithm}

\begin{theorem}[Correctness of Univariate Refinement Step]
\label{thm: uni_refinement}
When $\varepsilon \leqslant \frac{1}{12}$. Knowing that 
\begin{itemize}
    \item $\left|\expect{\RV{X}}\right| \leqslant \varepsilon$ \;
    \item $\Var{\RV{X}} \leqslant \left(\frac{1}{4}\right)^2$\;
\end{itemize}
Algorithm~\ref{alg: uni_refinement} return an mean estimate $\tilde{\mu}$ such that 
\begin{equation}
    \prob{\left|\tilde{\mu} - \expect{\RV{X}}\right| \leqslant \frac{1}{2} \varepsilon} \geqslant 1 - \delta
\end{equation}
\end{theorem}
\begin{proof}
By condition given we know that $\expect{\RV{X}^2} = \expect{\RV{X}}^2 + \Var{\RV{X}} \leqslant \left(\frac{\sqrt{10}}{12}\right)^2 < (\frac{1}{3})^2$, so Theorem~\ref{thm: alpha_RV} holds with $(\varepsilon, s_0)$ replaced with $(\varepsilon, \frac{\sqrt{10}}{12})$. Fix $\ell \in [M]$. There is an eigenvalue $e^{i \alpha}$ with eigenket $\ket{\alpha}$ where $\alpha \in (-\pi, \pi]$
\begin{itemize}
    \item $\left|\alpha - \expect{\RV{X}}\right| \leqslant 0.2403 \varepsilon \leqslant \frac{1}{4} \varepsilon$;
    \item $\left|\bra{\alpha}\ket{\bm{1}}\right|^2 \geqslant 0.9238$, where we used $\varepsilon \leqslant \frac{1}{12}$;
\end{itemize}
Using Theorem~\ref{thm: uni_phase_estimation} with $\kappa = 3$, set $N \geqslant \frac{24\pi}{\varepsilon}$ accordingly to algorithm, we then have $\frac{3}{N} \leqslant \frac{1}{8\pi}\varepsilon$, which means that:
\begin{equation}
    \begin{aligned}
        \prob{\left|\mu_\ell- \alpha\right| \leqslant \frac{1}{4}\varepsilon }  & = \prob{\left|\frac{\mu_\ell}{2\pi} - \frac{\alpha}{2\pi}\right| \leqslant \frac{1}{8\pi}\varepsilon } \geqslant \prob{\left|\frac{\mu_\ell}{2\pi} - \frac{\alpha}{2\pi}\right| \leqslant \frac{3}{N} } \\
        & \geqslant \left(1-\frac{1}{4}\right) \left|\bra{\alpha}\ket{\bm{1}}\right|^2 \geqslant 0.69285 > \frac{2}{3} \\
    \end{aligned}
\end{equation}
Since  $\left|\alpha - \expect{\RV{X}}\right| < \frac{1}{4} \varepsilon$ so we know that 
\begin{equation}
    \prob{\left|\mu_\ell - \expect{\RV{X}}\right| \leqslant \frac{1}{2} \varepsilon} \geqslant \frac{2}{3}
\end{equation}
By Theorem~\ref{thm: median}, the success probability is boosted to $1 - \delta$ by taking the median:
\begin{equation}
    \prob{\left|\tilde{\mu} - \expect{\RV{X}}\right| \leqslant \frac{1}{2} \varepsilon} \geqslant 1 - \delta
\end{equation}

\end{proof}

Following Remark~\ref{remark: alpha_RV}, we can also bound the complexity of the algorithm in terms of calls to the quantum experiment. We will also keep track of the number of phase estimation used for later convenience. 
\begin{remark}[Complexity of Univariate Refinement Step]
    \label{remark: uni_refinement_complexity}
    Algorithm~\ref{alg: uni_refinement} always uses $O\left(\frac{1}{\varepsilon} \log \frac{1}{\delta}\right)$ accesses of oracle, and runs phase estimation $\log \frac{1}{\delta}$ times.
\end{remark}

With Theorem~\ref{thm: uni_refinement}, we can repeatedly call Algorithm~\ref{alg: uni_refinement} for mean value estimation:
\begin{algorithm}
    \caption{\label{alg: uni_esimator}Constrained Univariate Mean Value Estimator}
    \KwData{Access to Quantum Experiment of Random Variable $\RV{X}$, variance bound $\sigma_0 \geqslant 0$, initial accuracy parameter $0 \leqslant \varepsilon_0 \leqslant \frac{1}{3} \sigma_0$, number of trials $n \in \mathbb{N}^+$, confidence parameter $0 < \delta < 1$}
    \KwResult{A mean value estimate $\tilde{\mu}$}
    $M \leftarrow \left\lceil\log_2 \frac{n \varepsilon_0}{\sigma_0}\right\rceil$ \;
    Set array $\varepsilon'$ of length $M$ with  $\varepsilon'_\ell = \frac{\varepsilon_0}{2^{\ell-1} 4 \sigma_0} \forall \ell \in [M]$\;
    Set array $\delta'$ of length $M$ via Theorem~\ref{thm: log_log} according to array $\varepsilon'$ with parameters $(\delta,R)$ set to $(\delta,2)$\;
    $\tilde{\mu} \leftarrow 0$ \;
    \For{$\ell = 1$ \KwTo $M$}{
        $\mu_\ell \leftarrow$ result of calling Algorithm~\ref{alg: uni_refinement} on $\frac{\RV{X}-\tilde{\mu}}{4 \sigma_0}$ with parameters $(\varepsilon, \delta)$ set to $(\varepsilon'_\ell, \delta'_\ell)$ \;
        $\tilde{\mu} \leftarrow \tilde{\mu} + 4 \sigma_0 \mu_\ell$\;
    }
    Output $\tilde{\mu}$\;
\end{algorithm}

\begin{theorem}[Correctness of Algorithm~\ref{alg: uni_esimator}]
\label{thm: uni_estimator}
Given random variable $\RV{X}$ such that 
\begin{itemize}
    \item $\Var{\RV{X}} \leqslant \sigma_0^2$\;
    \item $\expect{\RV{X}} \leqslant \varepsilon_0 \leqslant \frac{1}{3} \sigma_0$\;
\end{itemize}
Algorithm~\ref{alg: uni_esimator} give a mean estimate $\tilde{\mu}$ with
\begin{equation}
    \prob{\left|\tilde{\mu} - \expect{\RV{X}}\right|\leqslant \frac{\sigma_0}{n}} \geqslant 1 - \delta
\end{equation}
\end{theorem}
\begin{proof}
Define $\tilde{\mu}_\ell$ to be the value of $\tilde{\mu}$ after the $\ell$-th time Algorithm~\ref{alg: uni_refinement} is called. $\tilde{\mu}_0 = 0$. We can see that $\forall \ell \in [M]$, $\Var{\frac{\RV{X}-\tilde{\mu}_{\ell-1}}{4 \sigma_0}} \leqslant \left(\frac{1}{4}\right)^2$.
For the $\ell$-th time Algorithm~\ref{alg: uni_refinement} is called, if $\ell = 1$ define success condition as:
\begin{equation}
    P_1 \equiv \left(\left|\tilde{\mu}_1 - \expect{\RV{X}}\right| \leqslant 2 \sigma_0 \varepsilon'_1\right)
\end{equation}
For $\ell > 1$, define success condition as the following:
\begin{equation}
    P_\ell \equiv \begin{cases}
        \left|\tilde{\mu}_\ell - \expect{\RV{X}}\right| \leqslant 2 \sigma_0 \varepsilon'_\ell & \left(\left|\tilde{\mu}_{\ell-1} - \expect{\RV{X}}\right| \leqslant 2 \sigma_0 \varepsilon'_{\ell-1}\right) \\
        \text{true} & \text{otherwise} \\ 
    \end{cases}
\end{equation}
In other words, we define ``succeed'' as only discriminating when the previous iteration satisfy the desired constraint, otherwise the algorithm always ``succeed''. We observe two key properties:
\begin{itemize}
    \item By assumption on $\RV{X}$, $\expect{\frac{\RV{X}}{4\sigma_0}} \leqslant \frac{\varepsilon_0}{4 \sigma_0} \leqslant \frac{1}{12}$, so Theorem~\ref{thm: uni_refinement} holds with parameters $(\varepsilon, \delta)$ replaced with $( \frac{\varepsilon_0}{4\sigma_0}, \delta'_1)$, so we know that 
    \begin{equation}
        \mathbb{P}[P_1] = \prob{\left|\tilde{\mu}_1 - \expect{\RV{X}}\right| \leqslant 2\sigma_0 \varepsilon_0 } = \prob{\left|\mu_1 - \expect{\frac{\RV{X}}{4\sigma_0}}\right| \leqslant \frac{\varepsilon_0}{2} } \geqslant 1 - \delta'_1 
    \end{equation}
    \item Consider $\ell \in [M-1]$. Assume $\left|\tilde{\mu}_\ell - \expect{\RV{X}}\right| \leqslant 2 \sigma_0 \varepsilon'_\ell$. Then $\left|\frac{\expect{\RV{X}} - \tilde{\mu}_\ell}{4\sigma_0}\right| \leqslant \frac{\varepsilon'_\ell}{2} = \varepsilon'_{\ell+1} \leqslant \frac{1}{12}$. So Theorem~\ref{thm: uni_refinement} holds with parameters $(\varepsilon, \delta)$ replaced with $(\varepsilon'_{\ell+1}, \delta'_{\ell+1})$, so we know that 
    \begin{equation}
        \begin{aligned} 
            \prob{\left|\tilde{\mu}_{\ell+1} - \expect{\RV{X}}\right| \leqslant 2\sigma_0 \varepsilon'_{\ell+1} } &   = \prob{\left|4 \sigma_0 \mu_{\ell+1} - \expect{\RV{X} - \tilde{\mu}_\ell}\right|\leqslant 2\sigma_0 \varepsilon'_{\ell+1} } \\
                & = \prob{\left|\mu_{\ell+1} - \expect{\frac{\RV{X}-\tilde{\mu}_\ell}{4\sigma_0}}\right| \leqslant \frac{\varepsilon'_{\ell+1}}{2} } \geqslant 1 - \delta'_{\ell+1}
            \end{aligned} 
    \end{equation}
\end{itemize}
From the definition of our success condition $\{P_\ell\}$, we know that at $\ell$-th time Algorithm~\ref{alg: uni_refinement} is called it always succeeds with probability at least $1 - \delta_\ell$. Thus, we can use Theorem~\ref{thm: log_log} to know that
\begin{equation}
    \prob{P_1 \wedge P_2 \wedge \cdots P_M} \geqslant 1 - \delta
\end{equation}
But $P_1 \wedge P_2 \wedge \cdots P_M$ implies 
\begin{equation}
    \left|\tilde{\mu} - \expect{\RV{X}}\right| \leqslant 2 \sigma_0 \varepsilon'_M = \frac{\varepsilon_0}{2^M} \leqslant \frac{\sigma_0}{n}
\end{equation}
That is, we find that:
\begin{equation}
    \prob{\left|\tilde{\mu} - \expect{\RV{X}}\right|\leqslant \frac{\sigma_0}{n}} \geqslant 1 - \delta
\end{equation}
\end{proof}

We can also bound the complexity of algorithm, regardless of whether the input variable satisfies the nice constraint or not:
\begin{theorem}[Complexity of Algorithm~\ref{alg: uni_esimator}]
\label{thm: uni_complexity}
The algorithm always uses $O\left(n \log \frac{1}{\delta}\right)$ access to quantum experiments to output the result. It also does this in $O\left(\ceil{\log \frac{n\varepsilon_0}{\sigma_0}}\left(\log \ceil{\log \frac{n\varepsilon_0}{\sigma_0}} + \log \frac{1}{\delta}\right)\right)$ phase estimations.
\end{theorem}
\begin{proof}
We know that $\frac{1}{\varepsilon'_M} = \frac{2^{M-1}4 \sigma_0}{\varepsilon_0} \in O\left(n\right)$. With the complexity statement in Remark~\ref{remark: uni_refinement_complexity}, we know that the $\ell$-th call runs in $O\left(\frac{1}{\varepsilon'_\ell} \log \frac{1}{\delta'_\ell}\right)$ time and calls phase estimation $O\left(\log \frac{1}{\delta'}\right)$ times. Meanwhile, by Theorem~\ref{thm: log_log}, with $M \in O\left(\log \ceil{\frac{n\sigma_0}{\varepsilon_0}}\right)$, we thus confirms that the algorithm uses $O\left(\frac{1}{\varepsilon'_M} \log \frac{1}{\delta}\right) = O\left(n \log \frac{1}{\delta}\right)$ calls to the quantum experiment and invokes phase estimation $O\left(M \log \frac{M}{\delta}\right) = O\left(\ceil{\log \frac{n\varepsilon_0}{\sigma_0}}\left(\log \ceil{\log \frac{n\varepsilon_0}{\sigma_0}} + \log \frac{1}{\delta}\right)\right)$ times.
\end{proof}

Furthermore, by going through the steps of Algorithms~\ref{alg: uni_refinement} and \ref{alg: uni_esimator} we can easily verify the following:
\begin{remark}[Hybird Circuit]
\label{remark: uni_hybird}
Algorithm~\ref{alg: uni_esimator} is a bybird circuit satisfying the definition made in Theorem~\ref{thm: c_to_q}.
\end{remark}

In Algorithm~\ref{alg: uni_esimator} there is still an initial bound on the expectation of $\RV{X}$, but we can get rid of this constraint easily, by Remark~\ref{remark: quantum_ex_to_classical}, by kickstarting the algorithm with a classical estimator:
\begin{algorithm}
    \caption{\label{alg: not_so_uni_esimator}(Not-so) Constrained Univariate Mean Value Estimator}
    \KwData{Access to Quantum Experiment of Random Variable $\RV{X}$, variance bound $\sigma_0 \geqslant 0$, number of trials $n \in \mathbb{R}^+$, confidence parameter $0 < \delta < 1$}
    \KwResult{A mean value estimate $\tilde{\mu}$}
    Run Classical Mean Value Estimator in Theorem~\ref{thm: classical_uni_estimator} for $\RV{X}$ with parameters $(n, \delta)$ replaced with $\left(9, \frac{\delta}{2}\right)$, store as $\mu'$\;  
    Run Algorithm~\ref{alg: uni_esimator} on $\RV{X} - \mu'$ with parameters $(n, \sigma_0, \varepsilon_0,\delta)$ replaced with $\left(n, \sigma_0, \frac{1}{3} \sigma_0, \frac{\delta}{2}\right)$, let it be $\mu''$\;
    Output $\tilde{\mu} = \mu' + \mu''$\;
\end{algorithm}

\begin{theorem}[Algorithm~\ref{alg: not_so_uni_esimator}]
\label{thm: notso_uni_estimator}
For univariate random variable $\RV{X}$ such that $\Var{\RV{X}} \geqslant \sigma_0^2$, the algorithm returns an estimate $\tilde{\mu}$ with 
\begin{equation}
    \prob{\left|\tilde{\mu} - \expect{\RV{X}}\right|\leqslant \frac{\sigma_0}{n}} \geqslant 1 - \delta
\end{equation}
The algorithm uses $O\left(n \log \frac{1}{\delta}\right)$ accesses to the quantum experiment, and it calls phase estimation $O\left(\log n \left( \log \log n + \log  \frac{1}{\delta}\right)\right)$ times.
\end{theorem}
\begin{proof}

By Theorem~\ref{thm: classical_uni_estimator}, we know that with probability at least $1 - \frac{\delta}{2}$,
\begin{equation}
\left|\expect{\RV{X} - \mu'}\right| - \left|\mu' \RV{X}\right| \leqslant \frac{1}{3} \sigma_0
\end{equation}
If this happens then with probability at least $1 - \frac{\delta}{2}$, by Theorem~\ref{thm: uni_estimator} we find 
\begin{equation}
\label{eq: notso_good}
\left|\tilde{\mu} - \expect{\RV{X}} \right| = \left|\mu'' - \expect{\RV{X} - \mu'}\right| \leqslant \frac{\sigma_0}{n}
\end{equation}

Combined, Eq.~(\ref{eq: notso_good}) happens with probability at least $1 - \delta$ via union bound. By Theorem~\ref{thm: uni_complexity} it is also easy to see that the complexity is $O(n)$ and $O\left(\log n \left( \log \log n + \log  \frac{1}{\delta}\right)\right)$ in terms of quantum experiments and phase estimations, respectively. 
\end{proof}

Continuing on from here, we can then upgrade to handle a univariate random variable with unknown $\sigma$. Luckily for us, these classical reductions are already discussed in Ref.~\cite{Kothari_2022}. In the text, the authors presented a sequence of problems, each upgrading from the previous ones, eventually building up to handle a random variable with unknown deviation. Algorithm~\ref{alg: not_so_uni_esimator} achieves exactly the same result as Problem 6 in Ref.~\cite{Kothari_2022}. As a special case, we can drop back to a simpler setting corresponding to Problem 3 in the text.
\begin{corollary}[Problem 3 in Ref.~\cite{Kothari_2022}]
\label{cry: connect_to_Ryan}
Given random variable $\RV{X}$ with $\expect{\RV{X}^2} \leqslant 1$ (which means that $\Var{\RV{X}} \leqslant 1$, given parameter $\varepsilon \in (0,1)$, by running Algorithm~\ref{alg: not_so_uni_esimator} with parameters $(n, \delta)$ replaced with $\left(\frac{1}{\varepsilon}, \delta\right)$, we can use $O\left(\frac{1}{\varepsilon}\log \frac{1}{\delta}\right)$ to generate a mean estimate $\tilde{\mu}$ with
\begin{equation}
    \prob{\left|\tilde{\mu} - \expect{\RV{X}}\right|\leqslant \varepsilon } \geqslant 1 - \delta
\end{equation}
\end{corollary}
Corollary~\ref{cry: connect_to_Ryan} mirrors directly Problem 3 as found in Ref.~\cite{Kothari_2022}. Then we can follow the exact same steps to unlock solutions to all the other problems mentioned in the text. (This will be useful in Sec.~\ref{sec: final_classical_reduction}.) \footnote{As a comment, the reduction steps Ref.~\cite{Kothari_2022} have quite a bit free room for speedup by a constant factor. For example, for Problem 4, we can pick parameters such that $\bar{p}$ is reduced by $\frac{1}{2}$ each iteration instead of $\frac{3}{4}$. Additionally, we might get an ever-so-slight further speedup by using Algorithm~\ref{alg: refinement} directly instead of solution to Problem 3 in Ref.~\cite{Kothari_2022} as discussed in Corollary~\ref{cry: connect_to_Ryan}. Luckily for the authors of this current paper, these considerations can be swepted under our big-$O$ notations.}

A direct result of our discussion means that by Corollary~\ref{cry: connect_to_Ryan}, we have obtained a solution to the full univariate mean value estimation problem.
\begin{corollary}[Full Univariate Estimator]
\label{cry: uni_full_estimator}
Given accesses to quantum experiment for univariate random variable $\RV{X}$, using $O\left(n \log \frac{1}{\delta}\right)$ calls one can find an mean estimate $\tilde{\mu}$ such that:
\begin{equation}
    \prob{\left| \tilde{\mu} - \expect{\RV{X}}\right| \leqslant \frac{\sqrt{\Var{\RV{X}}}}{n}} \geqslant 1 - \delta
\end{equation}
\end{corollary}
Corollary~\ref{cry: uni_full_estimator} is not necessary for our paper, but it remains a noteworthy detour.

\section{Multivariate Mean Value Estimator}
\label{sec: multi_estimator}

In this section we establish the main results for this paper---two efficient algorithms for mean value estimation. 

\subsection{Simple Estimator}
\label{sec: simple_algo}

Consider a generic multivariate random variable $\vec{\RV{X}}$ with covariance matrix $\Sigma$. Before we start all the fun, first we establish a useful observation:

\begin{lemma}[From Multivariate Trace to Univariate Variance]
\label{lem: variance_prob_bound}
Let $G$ be hypercubic lattice of dimension $d$ and some resolution. Taking vector $\vec u$ from $G$ uniformly at random, we find that for $t > 0$:
\begin{equation}
    \mathbb{P}_{\vec u \sim G} \left[\Var{\expval{\vec u, \vec{\RV{X}}}} \geqslant t \right] \leqslant 2e^{-\frac{t}{D \tr \Sigma}} 
\end{equation}
$\mathbb{P}_{\vec u \in G}$ is the probability when taking $\vec u \in G$ uniformly at random, $D$ is some global constant whose exact value can be easily found. 
\end{lemma}
\begin{proof}
For all $\vec u \in G$, we find that:
\begin{equation}
    \Var \expval{\vec u, \vec{\RV{X}}} = \vec u^T \Sigma \vec u
\end{equation}
Since $\Sigma$ is non-negative, consider a set of orthonormal eigenvectors of $\Sigma$, the $j$-th of which we label $\vec v_j$ with eigenvalue $\lambda_j$. Then we know that:
\begin{equation}
\label{eq: var_formula}
 \Var \expval{\vec u, \vec{\RV{X}}} = \vec u^T \Sigma \vec u  = \sum_{j=1}^d \lambda_j \expval{\vec v_j, \vec u}^2
\end{equation}
Fix $j \in [d]$. We know that when taking $\vec u$ uniformly at random from $G$, $\mathbb{E}_{\vec u \sim G}\left[\expval{\vec v_j,\vec u}\right] = 0$. Since $\normtwo{\vec v_j} = 1$, we can apply Hoeffding's inequality to get:
\begin{equation}
    \mathbb{P}_{\vec u \sim G} \left[\left|\expval{\vec v_j, \vec u}\right| \geqslant x\right] \leqslant 2e^{-2x^2}
\end{equation}
for all $x > 0$, where we used the fact that $G \subseteq \left(-\frac{1}{2}, \frac{1}{2}\right)^d$. In statistical jargon, the variable $\expval{\vec v_j, \vec u}$ is sub-Gaussian. 

Now we have wrote the variable of interest as linear combination of squares of sub-Gaussian variables, we will use results on sub-Gaussian and sub-exponential random variables to refine our claims. Since definitions are inconsistent across different sources, we will use the definitions in Ref.~\cite{Vershynin_2018}. \footnote{Other sources include lecture notes as in Refs.~\cite{Rinaldo_2018, Rinaldo_2019}.} 

Note that each of $\expval{\vec v_j, \vec u}$ are sub-Gaussian with the parameter as defined in Proposition 2.5.2 in Ref.~\cite{Vershynin_2018} being some constant.Therefore, $\expval{\vec v_j, \vec u}^2$ is sub-exponential with the parameter defined as in Proposition 2.7.1 being a constant. Thus, for each $j$, for all $0 \leqslant g \leqslant \frac{1}{D_1}$, where $D_1 > 0$ is some constant, we find:
\begin{equation}
    \expect{e^{g\expval{\vec v_j, \vec u}^2}} \leqslant e^{D_1 g}
\end{equation}
Rewrote the above sentence such that the variable of interest is rescaled to $\lambda_j \expval{\vec v_j, \vec u}^2$, we found that 
\begin{equation}
    \forall  g \in \left[0,\frac{1}{D_1 \lambda_j}\right] \quad \expect{e^{g \lambda_j \expval{\vec v_j, \vec u}^2}} \leqslant e^{D_1\lambda_j g}
\end{equation}
Now, we define $a_j = \frac{\sum_{j=1}^d \lambda_j}{\lambda_j} = \frac{\tr \Sigma}{\lambda_j}$. This way we have $\sum_{j=1}^d \frac{1}{a_j} = 1$. Thus, we can apply H\"{o}lder's inequality generalized for multiple variables:
\begin{equation}
    \begin{aligned}
        \expect{e^{g \Var \expval{\vec u, \vec{\RV{X}}}}} & = \expect{e^{g \sum_{j=1}^d \lambda_j \expval{\vec v_j, \vec u}^2}} = \expect{\prod_{j=1}^d e^{g \lambda_j \expval{\vec v_j, \vec u}^2}} \\
        & \leqslant \prod_{j=1}^d \expect{\left(e^{g \lambda_j \expval{\vec v_j, \vec u}^2}\right)^{a_i}}^{\frac{1}{a_i}} = \prod_{j=1}^d \expect{e^{g a_j \lambda_j \expval{\vec v_j, \vec u}^2}}^{\frac{1}{a_i}} \\
        & \leqslant \prod_{j=1}^d \left(e^{D_1 \lambda_j g a_j}\right)^{\frac{1}{a_j}} = \prod_{j=1}^d e^{D_1 \lambda_j g} = e^{g D_1 \left(\sum_{j=1}^d \lambda_j\right)} = e^{g D_1 \tr \Sigma}
    \end{aligned}  
\end{equation}
We have not specified the range of $g$. First, $g \geqslant 0$. Then, we have picked $a_j$ such that for the above steps to work, for all $j$, $g a_j \leqslant \frac{1}{\lambda_j}$, i.e., $g \leqslant \frac{1}{\lambda_j a_j} = \frac{1}{\tr \Sigma}$. 

From here, it's easy to see that $\Var \expval{\vec u, \vec{\RV{X}}}$ is also sub-exponential with parameter $O(\tr \Sigma)$. \footnote{The lecture notes in \cite{Rinaldo_2018} seems to give a even better statement directly, but the author did not show any proof so we will ignore that statement.} Then we can use another equivalent definition for sub-exponential random variable to see that there is some constant $D$ such that:
\begin{equation}
    \mathbb{P}_{\vec u \sim G} \left[\Var{\expval{\vec u, \vec{\RV{X}}}} \geqslant t \right] \leqslant 2e^{-\frac{t}{D \tr \Sigma}} 
\end{equation}
The exact value for $D_1$ and hence $D$ can be evaluated from steps in Ref.~\cite{Vershynin_2018}. One might be able to further refine these parameters with more careful analysis. 
\end{proof}

Using similar ideas for the univariate mean value estimator, we construct a simple algorithm for estimating the mean. First, Algorithm~\ref{alg: refinement} attempts to refine our current knowledge of the mean.

\begin{algorithm}
\caption{\label{alg: refinement}Refinement Step of Simple Multivariate Mean Estimator}
\KwData{Access to Quantum Experiment for Random Variable $\vec{\RV{X}}$, accuracy parameter $\varepsilon \geqslant 0$, confidence parameter $0 < \delta < 1$}
\KwResult{A mean estimate $\vec{\tilde{\mu}}$}
$N \leftarrow 2^{\left\lceil \log_2 \frac{16\pi}{\varepsilon} \right\rceil}$\;
$M \leftarrow 2 \left\lceil \frac{18 \ln\frac{d}{\delta} - 1}{2} \right\rceil + 1$\;
Initialize an array of vectors $\vec{\mu'}$ of length $M$ \;
\For{$\ell \leftarrow 1$ \KwTo $M$}{
    Initialize state $\ket{\psi} = \frac{1}{N^{\frac{d}{2}}} \sum_{\vec u \in G} \ket{\vec u} \ket{\bm{1}} \quad \in \mathcal{H}_G \otimes \mathcal{H}_\Omega$ where $G$ is the hypercubic lattice of dimension $d$ and resolution $N$\;
    Construct $\mathcal{G}$, the Grover operator on random variable $2 \arctan\left(\frac{1}{2}\truncate{\expval{\vec u, \vec{\RV{X}}}}{\frac{1}{\lambda \varepsilon}}\right)$ where $\vec u$ is fetched from $\mathcal{H}_G$ as a control register and $\lambda = \frac{5}{4 - 5\left(\frac{\sqrt{10}}{360}d^\frac{1}{4}\right)^2}$ \;
    Run multidimensional phase estimation algorithm on $\mathcal{G}$ (i.e., replace $U^N$ in Algorithm~\ref{alg: phase_estimation} with $\mathcal{G}^N$) with $\mathcal{H}_\Omega$ register treated as ancilla. Let the result be $\vec x$. Append $2\pi \vec x$ to $\vec{\mu}'$, i.e., $\vec{\mu}'_\ell \leftarrow 2 \pi \vec x$\;
}
Initialize output vector $\vec{\tilde{\mu}}$ of dimension $d$\;
\For{$\alpha \leftarrow 1$ \KwTo $d$}{
    $\tilde{\mu}^\alpha \leftarrow \text{ median of } \left\{\ell \in [M]: \mu'^\alpha_\ell \right\}$\;
}
Output $\vec{\tilde{\mu}}$\;
\end{algorithm}

To clarify, in Algorithm~\ref{alg: refinement}, $\mathcal{G} = \mathcal{R}\mathcal{O}$ where $\mathcal{R}$ is the reflection gate on $\mathcal{H}_\Omega$ and $\mathcal{O}$ is a univariate phase oracle controlled by $\mathcal{H}_G$, as:
\begin{equation}
    \mathcal{O}\ket{\vec u}\ket{k} = e^{i \expval{\vec u, \vec{\mathcal{X}_k}}} \ket{\vec u} \ket{k}
\end{equation}
Thus we run phase estimation treating $\mathcal{H}_\Omega$ as an ancilla register. It can been seen that the multidimensional phase estimation in Algorithm~\ref{alg: phase_estimation} functions as usual with additional ancillas, except that in Theorem~\ref{thm: phase_estimation} the quantum states will be the full state of the system to account for entangled disturbances on ancilla. Here in Theorem~\ref{thm: refinement} we formally prove the validity of the algorithm.

\begin{theorem}[Correctness of Multivariate Refinement Step, Algorithm~\ref{alg: refinement}]
\label{thm: refinement}
Consider a multivariate random variable $\vec{\RV{X}}$ with covariance matrix $\Sigma$. When it satisfy
\begin{itemize}
    \item $\norm{\expect{\vec{\RV{X}}}}_\infty \leqslant \sqrt{d}\varepsilon$ \footnote{The reader can intuitively understand this as $\norm{\expect{\vec{\RV{X}}}}_\infty \leqslant \varepsilon$. We wrote $\norm{\expect{\vec{\RV{X}}}}_\infty \leqslant \sqrt{d}\varepsilon$ as a slight generalization, for later convenience for the proof of Theorem~\ref{thm: con_sim_estimator}.};
    \item $\tr \Sigma \leqslant \left(\frac{1}{120 d^{\frac{1}{4}}}\sqrt{\frac{1}{10D}}\right)^2$;
\end{itemize}
where $\varepsilon \leqslant \frac{1}{900 d^{\frac{3}{4}}}$, $D$ is the constant in Lemma~\ref{lem: variance_prob_bound}.
We claim that the algorithm returns an estimate $\vec{\tilde{\mu}}$ such that
\begin{equation}
    \prob{\norm{\vec{\tilde{\mu}} - \expect{\vec{\RV{X}}} }_\infty \leqslant \frac{\varepsilon}{2}} \geqslant 1 - \delta 
\end{equation}
\end{theorem}
\begin{proof}
Let $G$ be hypercubic lattice of dimension $d$ with resolution $N$ as specified in Algorithm~\ref{alg: refinement}.

As just discussed in the main text, Theorem~\ref{thm: phase_estimation} still holds, except that the quantum states involved are now the full state of the system with extra registers. When we run multidimensional phase estimation on $\mathcal{G}$ as defined in Algorithm~\ref{alg: refinement}, before the QFT step, we achieve a state
\begin{equation}
    \ket{\varphi'} = \frac{1}{N^{\frac{d}{2}}} \sum_{\vec u \in G} \ket{\vec{u}} \left(\mathcal{G}_{\vec u}^N \ket{\bm{1}}\right)
\end{equation}
where $\mathcal{G}_{\vec u}$ is the Grover operator for random variable $2 \arctan \left( \frac{1}{2} \expval{\vec u, \vec{\RV{X}}} \right)$ with each $\vec u \in G$. We aim to approximate the following:
\begin{equation}
    \ket{\varphi} = \frac{1}{N^{\frac{d}{2}}} \sum_{\vec u \in G} \ket{\vec{u}} \left(e^{i N\expval{\vec u, \expect{\vec{\RV{X}}}}} \ket{\bm{1}}\right)
\end{equation}
These states differ by:
\begin{equation}
    \label{eq: refinement_diff}
    \norm{\ket{\varphi'} - \ket{\varphi}}^2 = \frac{1}{|G|} \sum_{\vec u \in G} \norm{\mathcal{G}_{\vec u}^N \ket{\bm{1}} - e^{i N \expval{\vec u, \expect{\vec{\RV{X}}}}} \ket{1}}^2 
\end{equation}
as $|G| = N^d$. 

Now we obtain the following results:
\begin{itemize}
    \item By Lemma~\ref{lem: variance_prob_bound} we find that 
    \begin{equation}
        \mathbb{P}_{\vec u \in G}\left[\Var\expval{\vec u, \vec{\RV{X}}} \geqslant \left(\frac{1}{120d^\frac{1}{4}} \right)^2 \right] \geqslant 2 e^{-10}
    \end{equation}
    \item Using Hoeffding's inequality similar as in Lemma~\ref{lem: variance_prob_bound}, since $\frac{1}{|G|}\sum_{\vec u \in G} \expval{\vec u, \expect{\vec{\RV{X}}}} = 0$, we know that 
    \begin{equation}
        \label{eq: refinement_Hoeffding}
        \mathbb{P}_{\vec u \in G} \left[ \expval{\vec u, \expect{\vec{\RV{X}}}} \geqslant 2.5 \sqrt{d} \varepsilon \right] \geqslant 2 e^{- \frac{ 12.5 d \varepsilon^2}{d \varepsilon^2}} = 2 e^{-12.5}
    \end{equation}
    where we also note that $2.5 \sqrt{d} \varepsilon \leqslant \frac{1}{360d^\frac{1}{4}}$.
\end{itemize}

Using union bound, we know that for at least $1 - 2(e^{-10}+e^{-12}) \geqslant 1 - 9.83 \times 10^{-5}$ of all $\vec u \in G$, we must have:
\begin{itemize}
    \item $\expval{\vec u, \vec{\expect{\RV{X}}}} < 2.5 \sqrt{d} \varepsilon \leqslant \frac{1}{360d^\frac{1}{4}}$;
    \item $\Var\expval{\vec u, \vec{\RV{X}}} < \left(\frac{1}{120d^\frac{1}{4}}\right)^2$ which implies that $\expval{\vec u, \expect{\vec{\RV{X}}}}^2 < \left(\frac{1}{360d^\frac{1}{4}}\right)^2 + \left(\frac{1}{120d^\frac{1}{4}}\right)^2 = \left(\frac{\sqrt{10}}{360d^\frac{1}{4}}\right)^2$.
\end{itemize}
In this case, Theorem~\ref{thm: alpha_RV_diff} holds with $(\varepsilon, s_0)$ replaced by $\left(2.5 \sqrt{d} \varepsilon, \frac{\sqrt{10}}{360d^\frac{1}{4}}\right)$. So we find:
\begin{equation}
    \label{eq: refinement_bound_u_1}
    \norm{\mathcal{G}_{\vec u}^N \ket{\bm{1}} - e^{i \expval{\vec u, \expect{\vec{\RV{X}}}}} \ket{1}}^2 \leqslant 4.166 \times 10^{-3}
\end{equation}
where we used the fact that $N \leqslant 2 \times \frac{16 \pi}{\varepsilon}$, so we plug in $N \varepsilon \geqslant 36\pi \times 2.5 \sqrt{d}$.

When $\vec u$ is not nice enough and our assumptions no longer hold we use a generous bound:
\begin{equation}
    \label{eq: refinement_bound_u_2}
    \norm{\mathcal{G}_{\vec u}^N \ket{\bm{1}} - e^{i N \expval{\vec u, \expect{\vec{\RV{X}}}}} \ket{1}}^2 \leqslant 4
\end{equation}

Combine Eqs.~(\ref{eq: refinement_diff}, \ref{eq: refinement_bound_u_1}, and \ref{eq: refinement_bound_u_2}) we obtain:
\begin{equation}
    \label{eq: refinement_diff_bound}
    \norm{\ket{\varphi'}-\ket{\varphi}} \leqslant 4.166 \times 10^{-3} +9.83 \times 10^{-5} \times 4  \leqslant 4.56 \times 10^{-3} < \left(\frac{1}{12}\right)^2
\end{equation}

This result, combined with Theorem~\ref{thm: phase_estimation}, Corollary~\ref{cry: phase_estimation} more specifically, implies that for each $l$, at the end of phase estimation step, $\mu'_\ell$ satisfy
\begin{equation}
    \forall \alpha \in [d] \quad \mathbb{P} \left[\left|\frac{\left(\mu'_\ell\right)^\alpha}{2\pi} - \frac{\expect{\RV{X}^\alpha}}{2\pi}\right| > \frac{4}{N}\right] \leqslant \frac{1}{3}
\end{equation}

Because we have picked $N \geqslant \frac{16\pi}{\varepsilon}$ we know $\frac{4}{N} \leqslant \frac{1}{4\pi} \varepsilon$. Thus, for all $\alpha \in [d]$,
\begin{equation}
    \begin{multlined}
       \mathbb{P} \left[\left|\left(\mu'_\ell\right)^\alpha -\expect{\RV{X}^\alpha}\right| > \frac{1}{2} \varepsilon\right]  =  \mathbb{P} \left[\left|\frac{\left(\mu'_\ell\right)^\alpha}{2\pi} - \frac{\expect{\RV{X}^\alpha}}{2\pi}\right| > \frac{1}{4\pi} \varepsilon \right] \\
        \leqslant \mathbb{P}  \left[\left|\frac{\left(\mu'_\ell\right)^\alpha}{2\pi} - \frac{\expect{\RV{X}^\alpha}}{2\pi}\right| > \frac{4}{N} \right] \leqslant \frac{1}{3}
    \end{multlined}
\end{equation}
That is, the subroutine returns an accurate estimate with \begin{equation}
    \prob{\left|\left(\mu'_\ell\right)^\alpha -\expect{\RV{X}^\alpha}\right| \leqslant \frac{1}{2} \varepsilon} \geqslant \frac{2}{3}
\end{equation} 

By taking Theorem~\ref{thm: median}, with parameter $\delta$ replaced by our $\frac{\delta}{d}$ we obtain:
\begin{equation}
     \mathbb{P}[\left|\tilde{\mu}^\alpha - \expect{\RV{X}^\alpha} \right| > \frac{\varepsilon}{2}] \leqslant \frac{\delta}{d}
\end{equation}
Use union bound to combine all dimensions $\alpha \in [d]$, we conclude:
\begin{equation}
    \prob{\norm{\vec{\tilde{\mu}} - \expect{\vec{\RV{X}}}}_\infty \leqslant \frac{\varepsilon}{2}} \geqslant 1 - \delta 
\end{equation}

\end{proof}

Following Remark~\ref{remark: alpha_RV} we conclude
\begin{remark}[Complexity of Multivariate Refinement Step, Algorithm~\ref{alg: refinement}]
    \label{remark: sim_refinement_complexity}
    Algorithm~\ref{alg: refinement} uses $O\left(\frac{1}{\varepsilon}\log\frac{d}{\delta}\right)$ quantum experiments to $\vec{\RV{X}}$.
\end{remark}

In a fashion extremely similar to what was discussed in Section~\ref{sec: uni_estimator}, we construct a multivariate mean value estimator as found in Algorithm~\ref{alg: con_sim_estimator}.

\begin{algorithm}
\caption{\label{alg: con_sim_estimator} Constrained Simple Multivariate Mean Value Estimator}
\KwData{Access to Quantum Experiment to Multivariate Random Variable $\vec{\RV{X}}$, number of trials used $n \in \mathbb{N}^+$, variance ($\sqrt{\tr \Sigma}$) bound $\sigma_0 \geqslant 0$, mean bound $0 \leqslant \varepsilon_0 \leqslant \frac{2}{15}\sqrt{10D} \sigma_0$, confidence parameter $0 < \delta < 1$}
\KwResult{An mean value estimate $\vec{\tilde{\mu}}$}
$K \leftarrow 120 d^{\frac{1}{4}}\sqrt{10D}$\;
$M \leftarrow \left\lceil\log_2 \frac{n \varepsilon_0}{\sigma_0}\right\rceil$ \;
Set array $\varepsilon'$ of length $M$ with  $\varepsilon'_\ell = \frac{\varepsilon_0}{2^{\ell-1} K \sigma_0} \forall \ell \in [M]$\;
Set array $\delta'$ of length $M$ via Theorem~\ref{thm: log_log} according to array $\varepsilon'$ with parameters $(\delta,R)$ set to $(\delta,2)$\;
$\vec{\tilde{\mu}} \leftarrow 0$ \;
\For{$\ell = 1$ \KwTo $M$}{
    $\vec{\mu}_\ell \leftarrow$ result of calling Algorithm~\ref{alg: refinement} on $\frac{\vec{\RV{X}}-\vec{\tilde{\mu}}}{K \sigma_0}$ with parameters $(\varepsilon, \delta)$ set to $(\varepsilon'_\ell, \delta'_\ell)$ \;
    $\vec{\tilde{\mu}} \leftarrow \vec{\tilde{\mu}} + K \sigma_0 \vec{\mu}_\ell$\; 
}
Output $\tilde{\mu}$\;
\end{algorithm}
\begin{theorem}[Complexity of Algorithm~\ref{alg: con_sim_estimator}]
\label{thm: con_sim_estimator}
When $\vec{\RV{X}}$, with covariance matrix $\Sigma$, satisfy
\begin{itemize}
    \item $\tr \Sigma \leqslant \sigma_0$;
    \item $\normtwo{\expect{\vec{\RV{X}}}} \leqslant \varepsilon_0 \leqslant \frac{2}{15}\sqrt{10D} \sigma_0$;
\end{itemize}
Then Algorithm~\ref{alg: con_sim_estimator} returns a mean estimate $\vec{\tilde{\mu}}$ such that 
\begin{equation}
    \prob{\norm{\vec{\tilde{\mu}} - \expect{\vec{\RV{X}}}}_\infty \leqslant \frac{\sigma_0}{n}} \geqslant 1 - \delta
\end{equation}

\end{theorem}

\begin{proof}

The proof is nearly identical to that of Theorem~\ref{thm: uni_estimator} with only a few tweaks.

Define $\vec{\tilde{\mu}}_\ell$ to be the value of $\vec{\tilde{\mu}}$ after the $\ell$-th time Algorithm~\ref{alg: uni_refinement} is called. $\vec{\tilde{\mu}}_0 = 0$. 
$\forall \ell \in [M]$, consider random variable $\vec{\RV{Y}}_\ell = \frac{\vec{\RV{X}}-\vec{\tilde{\mu}}_{\ell-1}}{K \sigma_0}$ \footnote{$\ell$ is array index. There is your beloved index confusion.}, say it has covariance matrix $\Sigma_\ell$. We can see that $\tr \Sigma_\ell \leqslant \left(\frac{1}{120 d^{\frac{1}{4}}\sqrt{10D}}\right)^2$. 

For the $\ell$-th time Algorithm~\ref{alg: uni_refinement} is called, if $\ell = 1$ define success condition as:
\begin{equation}
    P_1 \equiv \left(\norm{\vec{\tilde{\mu}}_1 - \expect{\vec{\RV{X}}}}_\infty \leqslant \frac{K}{2} \sigma_0 \varepsilon'_1\right)
\end{equation}
Just as what was done in proof for Theorem~\ref{thm: uni_estimator}, for $\ell > 1$, define success condition as the following:
\begin{equation}
    P_\ell \equiv \begin{cases}
        \norm{\vec{\tilde{\mu}}_\ell - \expect{\vec{\RV{X}}}}_\infty \leqslant \frac{K}{2} \sigma_0 \varepsilon'_\ell & \left(\norm{\vec{\tilde{\mu}}_{\ell-1} - \expect{\vec{\RV{X}}}}_\infty \leqslant \frac{K}{2} \sigma_0 \varepsilon'_{\ell-1}\right) \\
        \text{true} & \text{otherwise} \\
    \end{cases} 
\end{equation}
This gives two key properties:
\begin{itemize}
    \item By assumption on $\vec{\RV{X}}$ and our specific pick of $K$, $\expect{\frac{\vec{\RV{X}}}{K\sigma_0}} \leqslant \frac{\varepsilon_0}{K \sigma_0} \leqslant \frac{1}{900d^{\frac{1}{4}}}$, so Theorem~\ref{thm: refinement} holds with parameters $(\varepsilon, \delta)$ replaced with $( \frac{\varepsilon_0}{K\sigma_0}, \delta'_1)$, so we know that 
    \begin{equation}
        \mathbb{P}[P_1] = \prob{\norm{\vec{\tilde{\mu}}_1 - \expect{\vec{\RV{X}}}}_\infty \leqslant \frac{K}{2} \sigma_0 \varepsilon'_1} = \prob{\norm{\vec{\mu}_1 - \expect{\frac{\vec{\RV{X}}}{K\sigma_0}}}_\infty \leqslant \frac{\varepsilon_0}{2} } \geqslant 1 - \delta'_1 
    \end{equation}
    \item Consider $\ell \in [M-1]$. Assume $\norm{\vec{\tilde{\mu}}_\ell - \expect{\vec{\RV{X}}}}_\infty \leqslant \frac{K}{2} \sigma_0 \varepsilon'_\ell$. Then $\norm{\frac{\expect{\vec{\RV{X}}} - \vec{\tilde{\mu}}_\ell}{K\sigma_0}}_\infty \leqslant \frac{\varepsilon'_\ell}{2} = \varepsilon'_{\ell+1}$ which means that $\normtwo{\frac{\expect{\vec{\RV{X}}} - \vec{\tilde{\mu}}_\ell}{K\sigma_0}} \leqslant \sqrt{d}\varepsilon'_{\ell+1} \leqslant \frac{1}{900d^{\frac{1}{4}}}$. So Theorem~\ref{thm: uni_refinement} holds with parameters $(\varepsilon, \delta)$ replaced with $(\varepsilon'_{\ell+1}, \delta'_{\ell+1})$, so we know that 
    \begin{equation}
        \begin{aligned}         
            \prob{\norm{\vec{\tilde{\mu}}_{\ell+1} - \expect{\vec{\RV{X}}}}_\infty \leqslant \frac{K}{2}\sigma_0 \varepsilon'_{\ell+1} } &   = \prob{\norm{K \sigma_0 \vec{\mu}_{\ell+1} - \expect{\vec{\RV{X}} - \vec{\tilde{\mu}}_\ell}}\infty\leqslant \frac{K}{2} \sigma_0 \varepsilon'_{\ell+1} } \\
                & = \prob{\norm{\vec{\mu}_{\ell+1} - \expect{\frac{\vec{\RV{X}}-\vec{\tilde{\mu}}_\ell}{K\sigma_0}}}_\infty \leqslant \frac{\varepsilon'_{\ell+1}}{2} } \geqslant 1 - \delta'_{\ell+1}
            \end{aligned} 
    \end{equation}
\end{itemize}
From the definition of our success condition $\{P_\ell\}$, we know that at $\ell$-th time Algorithm~\ref{alg: uni_refinement} is called it always succeeds with probability at least $1 - \delta_\ell$. Thus, we can use Theorem~\ref{thm: log_log} to know that
\begin{equation}
    \prob{P_1 \wedge P_2 \wedge \cdots P_M} \geqslant 1 - \delta
\end{equation}
But $P_1 \wedge P_2 \wedge \cdots P_M$ implies (due to our pick of $M$)
\begin{equation}
    \norm{\vec{\tilde{\mu}} - \expect{\vec{\RV{X}}}}_\infty \leqslant 2 \sigma_0 \varepsilon'_M = \frac{\varepsilon_0}{2^M} \leqslant \frac{\sigma_0}{n}
\end{equation}
That is, we find that:
\begin{equation}
    \prob{\norm{\vec{\tilde{\mu}} - \expect{\vec{\RV{X}}}}_\infty \leqslant \frac{\sigma_0}{n}} \geqslant 1 - \delta
\end{equation}
\end{proof}

The complexity is given by
\begin{theorem}[Complexity of Algorithm~\ref{alg: con_sim_estimator}]
    \label{thm: con_sim_complexity}
    The algorithm always uses $O\left(n d^\frac{1}{4} \log \frac{d}{\delta} \right)$ access to quantum experiment to return the result. In terms of quantum memory (outside all potential ancillas) it needs 1 register for results of quantum experiment. 
\end{theorem}
\begin{proof}
We have:
\begin{equation}
\frac{1}{\varepsilon'_M} = \frac{2^{M-1}K \sigma_0}{\varepsilon_0}  \in \frac{\sigma_0}{\varepsilon_0} \times O\left(\frac{n\varepsilon_0}{\sigma_0}\right) \times O\left( d^{\frac{1}{4}}\right) = O\left(n d^\frac{1}{4}\right)
\end{equation}
With the complexity statement in Remark~\ref{remark: sim_refinement_complexity} and Theorem~\ref{thm: log_log}, we thus confirms that the algorithm uses $ O\left(n d^\frac{1}{4} \log \frac{d}{\delta}\right)$ calls to the quantum experiment.

In terms of memory, we get to reuse the quantum registers whenever a quantum subroutine is called, so it's $O(1)$.
\end{proof}

Algorithm~\ref{alg: con_sim_estimator} gives an estimator without any $\log$ scalling on $n$. However, it scales badly with respect to $d$. The key reason is due to the extra $\sqrt{d}$ factor when we go from the accuracy parameter $\varepsilon$ for each dimension, to $2.5 \sqrt{d} \varepsilon$ as the accuracy parameter for the combined random variable $\expval{\vec u, \vec{\RV{X}}}$ at the Grover operator level. Since our Grover operator has an intrinsic uncertainty in the phase $\propto O(s_0^2 \varepsilon)$ where $\varepsilon$ is now replaced with $O(\sqrt{d} \varepsilon)$. We are thus forced to reduce $s_0$ to accommodate. 

\subsection{Meticulous Estimator}
\label{sec: mer_estimator}

The problems we just discussed also hints to a valid solution to get rid of the polynomial overhead in $d$. By converting the univariate mean value estimator in Algorithm~\ref{alg: uni_esimator} to a quantum circuit, one can obtain a unitary that approximates $e^{i N\expval{\vec u, \expect{\vec{\RV{X}}}}}$ arbitrarily close with few assumptions. 

By Remark~\ref{remark: uni_hybird} and Theorem~\ref{thm: c_to_q}, with $O(n)$ access to the quantum experiment of some univariate random variable $\RV{X}$, for some random variable $\RV{X}$, variance bound $\sigma_0 \geqslant 0$, mean bound $0 \leqslant \varepsilon_0 \leqslant \frac{1}{3}\sigma_0$, and $0 < \delta < 1$, we can construct unitary $U$ on $\mathcal{H}_\Omega' \otimes \mathcal{H}_\text{output}$ with:
\begin{equation}
    U\ket{0}\ket{0} = \sum_{j \in \Omega'} \sqrt{q_j} \ket{\lambda_j} \ket{y_j}
\end{equation}
where $\Omega'$ is some probability space corresponding measurements results during the execution of the algorithms, $y_j$ corresponds to instances of a random variable $\bm{y}$. Thus $U$ ``output'' a mean value estimate in the sense of the following lemma,
\begin{lemma}[Quantizing the Constrained Classical Univariate Estimator, Algorithm~\ref{alg: uni_esimator}]
\label{lem: c_to_q_uni}
For $\bm{y}, \Omega'$ that we just defined above, if the random variable $\RV{X}$ satisfy:
\begin{itemize}
    \item $\Var{\RV{X}} \leqslant \sigma_0^2$.
    \item $\expect{\RV{X}} \leqslant \varepsilon_0 \leqslant \frac{1}{3} \sigma_0$\; 
\end{itemize}
The unitary we constructed satisfy
\begin{equation}
   \mathbb{P}_{\Omega'}\left[\left|\bm{y} - \mathbb{E}[\RV{X}]\right| \leqslant \frac{\sigma_0}{n}\right] \geqslant 1 - \delta
\end{equation}
where $\mathbb{P}_{\Omega'}$ indicates that the probability is taken over $\Omega'$.
\end{lemma}

Given some integer $N$, we can also generate a phase operator $\mathcal{Q}$ that converts output register's value into a phase:
\begin{equation}
    \mathcal{Q}\ket{z} = e^{iNz} \ket{z}.
\end{equation}
From these unitaries we construct:
\begin{equation}
    V = U^\dag \left( I \otimes \mathcal{Q} \right) U
\end{equation}
The claim is that $V$ approximates a phase oracle to produce $e^{i N \expect{\RV{X}}}$, as found in the following theorem.
\begin{theorem}[Approximating the Phase Unitary using Quantum Subroutine]
\label{thm: uni_approx_phase}
Given some small constant $\xi > 0$, $N \in \mathbb{N}^+, \sigma_0 \geqslant 0$. We feed to $U$ with parameters $(n, \sigma_0, \varepsilon_0, \delta)$ replaced with $\left(\ceil{\frac{3N \sigma_0}{2\xi}}, \sigma_0, \varepsilon_0, \frac{1}{9}\xi^2\right)$, and we feed in $N$ to construction of $\mathcal{Q}$. 

Then, when the random variable $\RV{X}$ satisfy:
\begin{itemize}
    \item $\Var{\RV{X}} \leqslant \sigma_0^2$.
    \item $\expect{\RV{X}} \leqslant \varepsilon_0 \leqslant \frac{1}{3} \sigma_0$\; 
\end{itemize}
We find:
\begin{equation}
   \norm{V\ket{0}\ket{0} - e^{i N \expect{\RV{X}}} \ket{0} \ket{0}} \leqslant \xi
\end{equation}

\end{theorem}
\begin{proof}
Consider only the action of $(I \otimes \mathcal{Q}) U$, it produces a state:
\begin{equation}
\label{eq: uni_phase_proof_1}
\ket{\psi} = (I \otimes \mathcal{Q}) U \ket{0}\ket{0} = \sum_{j \in \Omega'} \sqrt{q_j} e^{i N y_j} \ket{\lambda_j} \ket{y_j}
\end{equation}
We want this state to be:
\begin{equation}
\label{eq: uni_phase_proof_2}
\ket{\varphi} = (I \otimes \mathcal{Q}) U \ket{0}\ket{0} = \sum_{j \in \Omega'} \sqrt{q_j} e^{i N \expect{\RV{X}}} \ket{\lambda_j} \ket{y_j} 
\end{equation}
We can see that $U^\dag \ket{\varphi} = e^{i N \expect{\RV{X}}} \ket{0}\ket{0}$. So:
\begin{equation}
\label{eq: uni_phase_proof_3}
\norm{V\ket{0}\ket{0} - e^{i N \expect{\RV{X}}} \ket{0} \ket{0}} = \norm{\ket{\psi} - \ket{\varphi}}
\end{equation}
Square on both side of Eq.~(\ref{eq: uni_phase_proof_3}) and plug in definitions in Eqs.~(\ref{eq: uni_phase_proof_1} and \ref{eq: uni_phase_proof_2}) we obtain:

\begin{equation}
\begin{aligned}
 \norm{V\ket{0}\ket{0} - e^{i N \expect{\RV{X}}} \ket{0} \ket{0}}^2 & = \norm{\sum_{j \in \Omega'} \sqrt{q_j} \left( e^{iN y_j} - e^{i N \expect{\RV{X}}}\right) \ket{\lambda_j} \ket{y_j}}^2 \\
& = \sum_{j \in \Omega'} q_j \left( e^{iN y_j} - e^{i N \expect{\RV{X}}}\right)^2
\end{aligned}
\end{equation}
where we used the fact that $\lambda_j$ are orthogonal vectors. 

By Lemma~\ref{lem: c_to_q_uni}, we know that for at least $1 - \delta$ of all times 
\begin{equation}
    \left|e^{iN y_j} - e^{i N \expect{\RV{X}}}\right| \leqslant N \left|y_j - \expect{\RV{X}}\right| \leqslant \frac{N \sigma_0}{n}
\end{equation}

For the other situations we can use a generous bound of 
\begin{equation}
    \left|e^{iN y_j} - e^{i N \expect{\RV{X}}}\right| \leqslant 2
\end{equation}
Combined we obtain
\begin{equation}
\hphantom{{}={}} \norm{V\ket{0}\ket{0} - e^{i N \expect{\RV{X}}} \ket{0} \ket{0}}^2 \leqslant 4 \delta + \left( \frac{N \sigma_0 }{n}\right)^2
\end{equation}

If we plug in $\delta = \frac{1}{9} \xi^2$ and $n \geqslant \frac{3N \sigma_0 }{2\xi}$, we will produce a bound:
\begin{equation}
\hphantom{{}={}} \norm{V\ket{0}\ket{0} - e^{i N \expect{\RV{X}}} \ket{0} \ket{0}}^2 \leqslant 2 \left(\frac{2}{3}\xi\right)^2 \leqslant \xi^2
\end{equation}
thus completing the proof. 
\end{proof}

But what's the cost of $V$? Here we see that while the time complexity becomes trivial, the space complexity will be dominated by the number of quantum subroutines being called, i.e., the number of phase estimations as what we have been tracking throughout Sec.~\ref{sec: uni_estimator}. 
\begin{theorem}[Cost of Phase Unitary from Quantizing Classical Univariate Estimator]
\label{thm: cost_V}
The unitary $V$ using parameters given in Theorem~\ref{thm: uni_approx_phase} costs $O\left(\ceil{\frac{N \sigma_0}{\xi}} \log \frac{1}{\chi}\right)$ in terms of accesses to quantum experiments. It also needs to allocate $O\left(\log \ceil{\frac{N \sigma_0}{\xi}} \left(\log \log \ceil{\frac{N \sigma_0}{\xi}} + \log \frac{1}{\xi}\right) \right)$ entangled quantum registers for quantum experiments alone. \footnote{Similarly as how we define big-$O$ notation with $\delta$ present in Footnote~\ref{footnote: O_with_delta}, we only care about the limit $\chi \to 0$ when using big-$O$ notation, this means that when appearing inside big-$O$, we may set, say, $\chi < c < 1$ for some constant $c$.}
\end{theorem}
\begin{proof}
By Remark~\ref{remark: uni_hybird} and Theorem~\ref{thm: c_to_q}, we can translate the complexities found in Theorem~\ref{thm: uni_complexity} accordingly to here. Notice that the quantum subrountine in the hybird algorithm is essentially phase estimation, so the number of phase estimations called translate into the number of registers needed for the quantized circuit. 
\end{proof}

With accesses to $V$ we can construct a multivariate mean value estimator by running phase estimation:
\begin{algorithm}
\caption{\label{alg: con_mer_estimator} Constrained Meticulous Multivariate Mean Estimator}
\KwData{Accesses to quantum experiment for random variable $\vec{\RV{X}}$, (not exactly) number of trails $n \in \mathbb{R}^+$, variance ($\sqrt{\Tr \Sigma}$) bound $\sigma_0 \geqslant 0$, mean bound $0 \leqslant \varepsilon_0 \leqslant \frac{2}{15}\sqrt{10D} \sigma_0$ \footnote{This parameter serves no purpose at all for this algorithm. We treat it as an input parameter to match up with Algorithm~\ref{alg: con_sim_estimator}.}, accuracy parameter $0 < \delta < 1$, assuming $n \geqslant \frac{\ln{\frac{d}{\delta}}}{\sqrt{\ln d}}$}
\KwResult{A mean estimate $\vec{\tilde{\mu}}$}
$K \leftarrow \sqrt{10D}\sigma_0$ \;
$N \leftarrow 2^{\ceil{\log_2 \left(8 \pi n\sqrt{10D}\right)}}$\;
$M \leftarrow 2 \left\lceil \frac{18 \ln\frac{d}{\delta} - 1}{2} \right\rceil + 1$\;
Initialize an array of vectors $\vec{\mu}'$ of length $M$\;
\For{$\ell \leftarrow 1$ \KwTo $M$}{
    Set unitary $V$ using parameters given in Theorem~\ref{thm: uni_approx_phase} but
    \begin{itemize}
        \item we use quantum experiments to random variable $\expval{\vec u, \frac{\vec{\RV{X}}}{K}}$ in place of $\RV{X}$, where $\vec u$ is fetched from $\mathcal{H}_G$ as a control register
        \item Replace parameters $\left(\sigma_0, N, \xi\right)$ with $\left(1, N, \frac{1}{13}\right)$\;
    \end{itemize}
    Run multidimensional phase estimation with controlled-$V$ in place of $U^N$ on register $\mathcal{H}_{G}$, where $G$ is hypercubic lattice of resolution $N$. Let the result be $\vec x$. Append $2\pi K \vec x$ to $\vec{\mu}'$, i.e., $\vec{\mu}'_\ell \leftarrow 2 \pi K \vec x$\;
}
Initialize output vector $\vec{\tilde{\mu}}$ of dimension $d$\;
\For{$\alpha \leftarrow 1$ \KwTo $d$}{
    $\tilde{\mu}^\alpha \leftarrow \text{ median of } \left\{\ell \in [M]: \mu'^\alpha_\ell \right\}$\;
}
Output $\vec{\tilde{\mu}}$\;

\end{algorithm}

\begin{theorem}[Correctness of Algorithm~\ref{alg: con_mer_estimator}]
\label{thm: con_mer_estimator}
When $\vec{\RV{X}}$, with covariance matrix $\Sigma$, satisfy
\begin{itemize}
    \item $\tr \Sigma \leqslant \sigma_0$;
    \item $\normtwo{\expect{\vec{\RV{X}}}} \leqslant \varepsilon_0 \leqslant \frac{2}{15}\sqrt{10D} \sigma_0$;
\end{itemize}
Then Algorithm~\ref{alg: con_mer_estimator} returns a mean estimate $\vec{\tilde{\mu}}$ such that 
\begin{equation}
    \prob{\norm{\vec{\tilde{\mu}} - \expect{\vec{\RV{X}}}}_\infty \leqslant \frac{\sigma_0}{n}} \geqslant 1 - \delta
\end{equation}
\end{theorem}
\begin{proof}
Similar to the proof in Theorem~\ref{thm: refinement}, when we run multidimensional phase estimation on with $V$ instead, before the QFT step, we achieve a state
\begin{equation}
    \ket{\psi} = \frac{1}{N^{\frac{d}{2}}} \sum_{\vec u \in G} \ket{\vec{u}} \left(\mathcal{V}_{\vec u} \ket{0}\ket{0}\right)
\end{equation}
where $\mathcal{V}_{\vec u}$ is unitary $V$ as discussed in main text except for random variable $\expval{\vec u, \frac{\vec{\RV{X}}}{K}}$. We wish to achieve a state:
\begin{equation}
    \ket{\varphi} = \frac{1}{N^{\frac{d}{2}}} \sum_{\vec u \in G} \ket{\vec{u}} \left(e^{i N \expval{\vec u, \frac{\vec{\RV{X}}}{K}}} \ket{0}\ket{0}\right)
\end{equation}
These states differ by:
\begin{equation}
    \label{eq: mer_diff}
    \norm{\ket{\varphi'} - \ket{\varphi}}^2 = \frac{1}{|G|} \sum_{\vec u \in G} \norm{V_{\vec u} \ket{0}\ket{0} - e^{i N \expval{\vec u, \expect{\frac{\vec{\RV{X}}}{K}}}}  \ket{0}\ket{0}}^2 
\end{equation}

Use Lemma~\ref{lem: variance_prob_bound} we know that:
\begin{equation}
    \mathbb{P}_{\vec u \sim G_N}\left[\Var{\expval{\vec u, \vec{\RV{X}}}} \geqslant \left(\sqrt{10D} \tr \Sigma\right)^2\right] \leqslant 2 e^{-10}
\end{equation}
From Hoeffding's inequality we also find:
\begin{equation}
    \mathbb{P}_{\vec u \sim G_N}\left[\left|\expect{\expval{\vec u, \vec{\RV{X}}}}\right| \leqslant \frac{1}{3}\sqrt{10D} \tr \Sigma\right] \leqslant 2 e^{-12.5}
\end{equation}

By union bound, this means that for probability at least $1 - 2 (e^{-10} + e^{-12.5})$ of all $\vec u \in G$ we have 
\begin{itemize}
    \item $\Var{\expval{\vec u, \vec{\RV{X}}}} < \left(\sqrt{10D} \tr \Sigma\right)^2$;
    \item $\left|\expect{\expval{\vec u, \vec{\RV{X}}}}\right| < \frac{1}{3}\sqrt{10D} \tr \Sigma$;
\end{itemize}
Our selection of $K$ ensures that:
\begin{itemize}
    \item $\Var{\expval{\vec u, \frac{\vec{\RV{X}}}{K}}} < 1$;
    \item $\left|\expect{\expval{\vec u, \frac{\vec{\RV{X}}}{K}}}\right| < \frac{1}{3}$;
\end{itemize}
So Theorem~\ref{thm: uni_approx_phase} applies with parameters $(\sigma_0, N, \chi)$ replaced with $(1, N, \frac{1}{13})$ (and $\RV{X}$ replaced with $\expval{\vec u, \frac{\vec{\RV{X}}}{K}}$ which gives 
\begin{equation}
    \norm{V_{\vec u} \ket{0}\ket{0} - e^{i N \expval{\vec u, \expect{\frac{\vec{\RV{X}}}{K}}}  \ket{0}\ket{0}}}^2 \leqslant \left(\frac{1}{13}\right)^2
\end{equation}

When the nice constraints are not satisfied, we use a generous bound of 
\begin{equation}
    \norm{V_{\vec u} \ket{0}\ket{0} - e^{i N \expval{\vec u, \expect{\frac{\vec{\RV{X}}}{K}}}}  \ket{0}\ket{0}}^2 \leqslant 4
\end{equation}
Add these results into Eq.~(\ref{eq: mer_diff}) we find:
\begin{equation}
\norm{\ket{\varphi'} - \ket{\varphi}}^2 \leqslant \left(\frac{1}{13}\right)^2 + 2\left(e^{-10}+e^{-12.5}\right) \times 4 < 6.311\times 10^{-3} < \left(\frac{1}{12}\right)^2 
\end{equation}

So Theorem~\ref{thm: phase_estimation} and Corollary~\ref{cry: phase_estimation} holds, which means that for each dimension $\alpha \in [d]$, for each $\ell \in [M]$, with probability at least $\frac{2}{3}$ we have:
\begin{equation}
\left|\frac{\left(\mu'_\ell\right)^\alpha}{2\pi K} - \frac{1}{2\pi}\expect{\frac{\RV{X}^\alpha}{K}}\right| \leqslant \frac{4}{N}
\end{equation}
which means that:
\begin{equation}
\left|\left(\mu'_\ell\right)^\alpha -\expect{\RV{X}^\alpha}\right| \leqslant \frac{8\pi K}{N} 
\end{equation}
Due to our pick of $N$ such that $N \geqslant \frac{8 \pi}{n}\sqrt{10D}$ we get:
\begin{equation}
    \label{eq: mer_proof}
    \left|\left(\mu'_\ell\right)^\alpha -\expect{\RV{X}^\alpha}\right| \leqslant \frac{\sigma_0}{n}
\end{equation}

By taking a coordinate-wise median, Theorem~\ref{thm: median}, we obtain $\vec{\tilde{\mu}}$ such that:
\begin{equation}
    \prob{\norm{\vec{\tilde{\mu}} -\expect{\vec{\RV{X}}}}_{\infty} \leqslant \frac{\sigma_0}{n}} \geqslant 1 - \delta
\end{equation}
\end{proof}

The complexity of the algorithm is given by
\begin{theorem}[Complexity of Algorithm~\ref{alg: con_mer_estimator}]
\label{thm: con_mer_complexity}
Algorithm~\ref{alg: con_mer_estimator} uses $O\left(n \log \frac{d}{\delta}\right)$ access to the quantum experiment for multivariate random variable $\vec{\RV{X}}$. In terms of quantum memory it costs $O\left(\log n \log \log n\right)$.
\end{theorem}
\begin{proof}
In Algorithm~\ref{alg: con_mer_estimator}, by Remark~\ref{remark: quantum_experiment_algebra} and Theorem~\ref{thm: cost_V}, each call to $V$ consumes $O(N) = O\left(n\right)$, as we called Theorem~\ref{thm: uni_approx_phase} but with $\sigma_0, \xi$ replaced with constants. We ran $V$ for $O\left(\log \frac{d}{\delta}\right)$ times so the total cost is $O\left(n \log \frac{d}{\delta}\right)$.

Similarly, the memory cost is $O\left(\log n \log \log n\right)$ for each $V$. This memory can be reused for each multidimensional phase estimation so the final cost is also $O\left(\log n \log \log n\right)$.

\end{proof}

\subsection{Final Classical Reduction}
\label{sec: final_classical_reduction}

In Theorem~\ref{alg: con_sim_estimator} and Theorem~\ref{alg: con_mer_estimator} we have constructed two different multivariate mean value estimators, each with their advantages and disadvantages. In both cases we need to give a bound of the variance ($\sqrt{\tr \Sigma}$) for the random variable and ensures that it has a decently small mean (we can do this by give initial estimate and shift the random variable). In this section we use these primitives to handle an arbitary multivariate random variable with no previous assumptions known. 

For the remainder of this section, let's use $\mathcal{A}$ as a placeholder for an algorithm depending on random variable $\vec{\RV{X}}$ and parameters $\left(n, \sigma_0, \varepsilon_0, \delta\right)$ as described in Theorem~\ref{thm: con_sim_estimator} and \ref{thm: con_mer_estimator}, such that when $\vec{\RV{X}}$, with covariance matrix $\Sigma$, satisfy:
\begin{itemize}
    \item $\tr \Sigma \leqslant \sigma_0$;
    \item $\normtwo{\expect{\vec{\RV{X}}}} \leqslant \varepsilon_0 \leqslant \frac{2}{15}\sqrt{10D} \sigma_0$;
\end{itemize}
$\mathcal{A}$ returns a mean estimate $\vec{\tilde{\mu}}$ such that 
\begin{equation}
    \prob{\norm{\vec{\tilde{\mu}} - \expect{\vec{\RV{X}}}}_\infty \leqslant \frac{\sqrt{\tr \Sigma}}{n}} \geqslant 1 - \delta
\end{equation}
It can be later substituted into either Algorithm~\ref{alg: con_sim_estimator} or \ref{alg: con_mer_estimator}.

First, similar to us elimiating the dependence on upper bound of $\expect{\RV{X}}$ on univariate case via Theorem~\ref{alg: not_so_uni_esimator}, we run a couple of classical trials first to complete an algorithm without dependence on $\expect{\vec{\RV{X}}}$:

\begin{algorithm}
    \caption{\label{alg: not_so_multi_esimator}(Not-so) Constrained Multivariate Mean Value Estimator}
    \KwData{Access to Quantum Experiment of Random Variable $\vec{\RV{X}}$, variance bound $\sigma_0 \geqslant 0$, number of trials parameter $n \in \mathbb{N}^+$, confidence parameter $0 < \delta < 1$}
    \KwResult{A mean value estimate $\vec{\tilde{\mu}}$}
    $n' \leftarrow \ceil{\frac{225}{4}10D\left(1 + \sqrt{\frac{1}{\ln \frac{2}{\delta}}}\right)^2} $ \footnote{The classical mean estimator primitive does need $n'$ to be integer so we have to take a ceil function}\;
    Run Classical Multivariate Mean Value Estimator in Theorem~\ref{thm: classical_multi_estimator} for $\vec{\RV{X}}$ with parameters $(n,\delta)$ replaced with $\left(n', \frac{\delta}{2}\right)$, store as $\vec{\mu}'$\;  
    Run $\mathcal{A}$ on $\vec{\RV{X}} - \vec{\mu}'$ with parameters $(n, \sigma_0, \varepsilon_0, \delta)$ replaced with $\left(n, \sigma_0, \frac{2}{15}\sqrt{10D} \sigma_0, \frac{\delta}{2}\right)$, let it be $\vec{\mu}''$\;
    Output $\vec{\tilde{\mu}} = \vec{\mu}' + \vec{\mu}''$\;
\end{algorithm}

\begin{theorem}[Analysis for Algorithm~\ref{alg: not_so_multi_esimator}]
\label{thm: notso_multi_estimator}
For multivariate random variable $\vec{\RV{X}}$ with covariance matrix $\Sigma$, such that $\tr \Sigma \leqslant \sigma_0^2$, the algorithm returns an estimate $\tilde{\mu}$ with 
\begin{equation}
    \prob{\norm{\vec{\tilde{\mu}} - \expect{\vec{\RV{X}}}}_\infty\leqslant \frac{\sigma_0}{n}} \geqslant 1 - \delta
\end{equation}
The algorithm uses $O\left(\log\frac{1}{\delta}\right)$ accesses to the quantum experiment along with a call to $\mathcal{A}$ with parameters $(n, \delta)$ replaced with $\left(O(n), \Omega(\delta)\right)$.
\end{theorem}
\begin{proof}

The proof is almost the same as Theorem~\ref{thm: notso_uni_estimator}. We have selected the value of $n'$ such that by Theorem~\ref{thm: classical_multi_estimator} and Corollary~\ref{cry: classical_multi_estimator_relaxed}, we know that with probability at least $1 - \frac{\delta}{2}$,
\begin{equation}
\normtwo{\expect{\vec{\RV{X}} - \vec{\mu}'}} = \normtwo{\vec{\mu}'-\vec{\RV{X}}} \leqslant \frac{2}{15}\sqrt{10D}  \sigma_0
\end{equation}
If this happens then with probability at least $1 - \frac{\delta}{2}$, by definition of $\mathcal{A}$ we find 
\begin{equation}
\label{eq: multi_notso_good}
\norm{\vec{\tilde{\mu}} - \expect{\vec{\RV{X}}}}_\infty = \norm{\vec{\mu}'' - \expect{\vec{\RV{X}} - \vec{\mu}'}}_\infty \leqslant \frac{\sigma_0}{n}
\end{equation}

Combined, Eq.~(\ref{eq: multi_notso_good}) happens with probability at least $1 - \delta$ via union bound. It is also easy to see that the complexity bound is justified.
\end{proof}

In Sec.~\ref{sec: uni_estimator}, by proving Corollary~\ref{cry: connect_to_Ryan} we have unlocked solutions to important problems mentioned in the classical reduction in Ref.~\cite{Kothari_2022}. These steps in the paper are designed to handle an unknown standard deviation, here we will do exactly the same for our multivariate random variable, but this time to handle an unknown $\tr \Sigma$. To build a complete chain of logic, we will only need to make slight modifications starting from Problem 8 in Ref.~\cite{Kothari_2022}. Formally speaking, first, by using the exact proof (adding steps to boost probability arbitrarily high) up to Problem 7 in Ref.~\cite{Kothari_2022} we can do the following:
\begin{lemma}[Problem 7 in Ref.~\cite{Kothari_2022}]
\label{lemma: Ryan_useful_lemma}
Given $O\left(n \log \frac{1}{\delta}\right)$ accesses to quantum experiment for univariate random variable $\RV{X}$ where $\forall k \in \Omega\; 0 \leqslant \mathcal{X}_k \leqslant 1$, there is an algorithm that returns an estimate for the mean $\tilde{\mu}$ such that
\begin{equation}
    \prob{\left|\tilde{\mu} - \expect{\RV{X}}\right| \leqslant \frac{\sqrt{\expect{\RV{X}}}}{n}} \geqslant 1 - \delta 
\end{equation}
\end{lemma}

Now following almost exact steps starting from Problem 8 in Ref.~\cite{Kothari_2022}, we can give an algorithm that computes the mean with no prior assumption on the random variable itself:
\begin{algorithm}
    \caption{\label{alg: final_multi_esimator}Multivariate Mean Value Estimator}
    \KwData{Access to Quantum Experiment of Random Variable $\vec{\RV{X}}$, number of trials $n \in \mathbb{N}^+$, confidence parameter $0 < \delta < 1$}
    \KwResult{A mean value estimate $\vec{\tilde{\mu}}$}
    $p \leftarrow \frac{25}{52C n^2}$, where $C$ is the parameter discussed in Quantum Estimation in Theorem~\ref{thm: quantile}\;
    $n' \leftarrow  2 \sqrt{\frac{52}{25}} n$\;
    Run Classical Multivariate Mean Value Estimator in Theorem~\ref{thm: classical_multi_estimator} for $\vec{\RV{X}}$ with parameters $(n, \delta)$ replaced with $\left(\ceil{25\left(1 + \sqrt{\frac{1}{\ln\frac{4}{\delta}}}\right)^2},\sigma_0, \frac{\delta}{4}\right)$, store as $\vec{\mu}'$\;  
    Run quantile estimation (Theorem~\ref{thm: quantile}) on $\normtwo{\vec{\RV{X}} - \vec{\mu}'}$ with parameters $(p, \delta)$ replaced with $\left(p, \frac{\delta}{4}\right)$, let result be $K$\;
    Let $\vec{\RV{Y}} = \truncate{\vec{\RV{X}} - \vec{\mu}'}{K}$, run algorithm discussed in Lemma~\ref{lemma: Ryan_useful_lemma} on random variable $\left(\frac{\normtwo{\vec{\RV{Y}}}}{K}\right)^2$ with parameters $(n,\delta)$ replaced with $\left(\frac{3}{\sqrt{p}}, \frac{\delta}{4}\right)$, let the result be $s'^2$ with $s' \geqslant 0$\;
    Run Algorithm~\ref{alg: not_so_multi_esimator} we just discussed on quantum experiment for $\vec{\RV{Y}}$ with parameters $(n, \sigma_0, \delta)$ replaced with $(n', \sqrt{\frac{3}{2}}Ks', \frac{\delta}{4})$, say result is $\vec{\mu}''$\;
    Output $\vec{\tilde{\mu}} = \vec{\mu}' + \vec{\mu}''$\;
\end{algorithm}
\begin{theorem}[Generic Multivariate Estimator Main Result]
\label{thm: generic_multi_main}
Algorithm~\ref{alg: final_multi_esimator} uses $1$ calls to $\mathcal{A}$ with parameters $(n, \delta)$ replaced with $\left(O(n), \Omega(\delta)\right)$, along with other procedures that costs 
\begin{itemize}
    \item $O\left(n \log \frac{1}{\delta}\right)$ in terms of quantum experiments for $\vec{\RV{X}}$ (with covariance matrix $\Sigma$)
    \item $O(1)$ in terms of quantum registers needed
\end{itemize}  
to return a mean value estimate with
\begin{equation}
    \prob{\norm{\vec{\tilde{\mu}} - \expect{\vec{\RV{X}}}}_\infty \leqslant \frac{\sqrt{\tr \Sigma}}{n}} \geqslant 1 - \delta
\end{equation}
which implies 
\begin{equation}
    \prob{\norm{\vec{\tilde{\mu}} - \expect{\vec{\RV{X}}}}_2 \leqslant \frac{\sqrt{d \tr \Sigma}}{n}} \geqslant 1 - \delta
\end{equation}
\end{theorem}
\begin{proof}
Assuming all calls succeed. Define $\RV{\vec{Z}} = \vec{\RV{X}}-\vec{\mu'}$. Then we set the parameters to the the initial call to classical mean value estimator such that:
\begin{equation}
    \label{eq: multi_estimator_proof_1}
    \normtwo{\expect{\vec{\RV{Z}}}} = \normtwo{\expect{\vec{\RV{X}}}-\vec{\mu'}} \leqslant \frac{\sqrt{\tr \Sigma}}{5}
\end{equation}
Then quantile estimation returns some $K$ with:
\begin{equation}
    p \leqslant \prob{\normtwo{\vec{\RV{Z}}} \geqslant K} \leqslant Cp
\end{equation}
where $C$ is the parameter discussed in Quantum Estimation in Theorem~\ref{thm: quantile}. Thus we know that:
\begin{equation}
    \expect{\normtwo{\vec{\RV{Y}}}^2} \geqslant K^2 \times \prob{\normtwo{\vec{\RV{Z}}} \geqslant K} \geqslant K^2p
\end{equation}
Thus $p \leqslant \frac{\expect{\normtwo{\vec{\RV{Y}}}^2}}{K^2}$. By choosing the parameter for algorithm in Lemma~\ref{lemma: Ryan_useful_lemma} to be $\frac{3}{\sqrt{p}}$, we have obtained:
\begin{equation}
\left|s'^2-\frac{\expect{\normtwo{\vec{\RV{Y}}}^2}}{K^2}\right| \leqslant \frac{1}{3K}\sqrt{p\expect{\normtwo{\vec{\RV{Y}}}^2}} \leqslant \frac{1}{3} \frac{\expect{\normtwo{\vec{\RV{Y}}}^2}}{K^2}
\end{equation}
In other words,
\begin{equation}
\label{eq: multi_estimator_proof_2}
    \left|K^2s'^2-\expect{\normtwo{\vec{\RV{Y}}}^2}\right| \leqslant \frac{1}{3} \expect{\normtwo{\vec{\RV{Y}}}^2}
\end{equation}
Thus we know that $\frac{3}{2} K^2 s'^2 \geqslant \expect{\normtwo{\vec{\RV{Y}}}^2}$ is guarenteed to be true. 

Let the covariance matrix of $\vec{\RV{Y}}$ to be $\Sigma'$ thus $\tr \Sigma' = \expect{\normtwo{\vec{\RV{Y}} - \expect{\vec{\RV{Y}}}}^2} \leqslant \expect{\normtwo{\vec{\RV{Y}}}^2} \leqslant  \frac{3}{2} K^2 s'^2$. So we know that Algorithm~\ref{alg: not_so_multi_esimator} is successful and Theorem~\ref{thm: notso_multi_estimator} applies, which means we find mean estimate $\mu''$ with 
\begin{equation}
    \norm{\vec{\mu}'' - \expect{\vec{\RV{Y}}}}_\infty \leqslant \frac{\sqrt{3 / 2} K s'}{ n'}
\end{equation}

Now we would like an upper bound in the form of $\frac{3}{2} K^2 s'^2 \leqslant O\left( \tr \Sigma\right)$. Eq.~(\ref{eq: multi_estimator_proof_2}) also implies that $K^2 s'^2 \leqslant \frac{4}{3} \expect{\normtwo{\vec{\RV{Y}}}^2}$ so $ \frac{3}{2} K^2 s'^2 \leqslant 2 \expect{\normtwo{\vec{\RV{Y}}}^2}$. At the start of proof we found Eq.~(\ref{eq: multi_estimator_proof_1}) which gives:
\begin{equation}
    \label{eq: multi_estimator_proof_3}
    \expect{\normtwo{\vec{\RV{Y}}}^2} \leqslant \expect{\normtwo{\vec{\RV{Z}}}^2} = \tr \Sigma + \normtwo{\expect{\vec{\RV{Z}}}}^2 \leqslant \frac{26}{25} \tr \Sigma 
\end{equation}
Because $\vec{\RV{Z}}$'s covariance matrix is also $\tr \Sigma$. So $\frac{3}{2} K^2 s'^2 \leqslant \frac{52}{25} \tr \Sigma$. Our pick of $n' \geqslant 2 \sqrt{\frac{52}{25}} n$ ensures that:
\begin{equation}
\norm{\vec{\mu}'' - \expect{\vec{\RV{Y}}}}_\infty \leqslant \frac{\sqrt{\tr \Sigma}}{ 2 n}
\end{equation}
Equivalently:
\begin{equation}
\label{eq: multi_estimator_proof_4}
\norm{\left(\vec{\mu}' + \vec{\mu}''\right) - \left(\expect{\vec{\RV{Y}}} + \vec{\mu}'\right)}_\infty \leqslant \frac{\sqrt{\tr \Sigma}}{ 2 n}
\end{equation}

Now we would like to bound the deviation between $\expect{\vec{\RV{Y}}} + \vec{\mu}'$ and $\expect{\vec{\RV{X}}}$. Remember that $\vec{\RV{Y}} = \truncate{\vec{\RV{Z}}}{K}$.By Cauchy Schwarz inequality, for all $\alpha \in [d]$ 
\begin{equation}
    \expect{\left|\left(\RV{Y}^\alpha + \mu'^\alpha\right) - \RV{X}^\alpha \right|} = \expect{\left|\RV{Y}^\alpha - \RV{Z}^{\alpha}\right|} \leqslant \sqrt{\expect{\left(\RV{Z}^\alpha\right)^2} \prob{\normtwo{\vec{\RV{Z}}}\geqslant K}} \leqslant \sqrt{\expect{\left(\RV{Z}^\alpha\right)^2} Cp} 
\end{equation}
Relaxing this bound, we obtain:
\begin{equation}
    \norm{\expect{\left(\vec{\RV{Y}} + \vec{\mu}'\right) - \vec{\RV{X}}}_\infty} \leqslant \sqrt{\expect{\normtwo{\vec{\RV{Z}}}^2} C p} \leqslant \sqrt{\frac{26}{25} Cp \tr \Sigma}
\end{equation}
where we have evoked part of Eq.~(\ref{eq: multi_estimator_proof_3}). By our pick of $p = \frac{25}{52C n^2}$, we obtain:
\begin{equation}
    \norm{\expect{\left(\vec{\RV{Y}} + \vec{\mu}'\right) - \vec{\RV{X}}}_\infty} \leqslant \frac{\sqrt{\tr \Sigma}}{ 2 n}
\end{equation}
Combined with Eq.~(\ref{eq: multi_estimator_proof_4}) we find:
\begin{equation}
\label{eq: multi_estimator_proof_5}
\norm{\left(\vec{\mu}' + \vec{\mu}''\right) - \expect{\vec{\RV{X}}}}_\infty \leqslant \frac{\sqrt{\tr \Sigma}}{ n}
\end{equation}
All the above calls to the algorithms succeed with probability at least $1- \delta$ by union bound, so Eq.~(\ref{eq: multi_estimator_proof_5}) happens with probability at least $1 - \delta$. Thus we have reached our conclusion.

Calling classical estimator costs $O\left(\log \frac{1}{\delta}\right)$, quantile estimation costs $O\left(\frac{\log{\frac{4}{\delta}}}{\sqrt{p}}\right) = O\left(n \log \frac{1}{\delta}\right)$. Run algorithm in Lemma~\ref{lemma: Ryan_useful_lemma} also costs $O\left(n \log\frac{1}{\delta} \right)$. Calling Algorithm~\ref{alg: not_so_multi_esimator} costs $O\left(\log \frac{1}{\delta}\right)$ along with a call to $\mathcal{A}$ with parameters $(n, \delta)$ replaced with $\left(O(n), \Omega(\delta)\right)$. In terms of memory, one quantum register suffices for the quantum experiment as we can always reuse it. Combined the total cost is $O\left(n \log \frac{1}{\delta}\right)$ in time and $O(1)$ in memory, along with whatever needed for a call to $\mathcal{A}$ with parameters $(n, \delta)$ replaced with $\left(O(n), \Omega(\delta)\right)$.
\end{proof}

Now let's simply substituite Algorithm~\ref{alg: con_sim_estimator} and \ref{alg: con_mer_estimator} for our placeholder algorithm $\mathcal{A}$ and we obtain the final theorems for the paper. 

\begin{theorem}[Main Result for Simple Estimator]
\label{thm: sim_main}
Given quantum experiments for a multivariate random variable $\vec{\RV{X}}$ with unknown covariance matrix $\Sigma$, there is an efficient algorithm that outputs a mean estimate $\vec{\tilde{\mu}}$ such that:
\begin{equation}
    \prob{\norm{\vec{\tilde{\mu}} - \expect{\vec{\RV{X}}}}_\infty \leqslant \frac{\sqrt{\tr \Sigma}}{n}} \geqslant 1 - \delta
\end{equation}
which implies 
\begin{equation}
    \prob{\norm{\vec{\tilde{\mu}} - \expect{\vec{\RV{X}}}}_2 \leqslant \frac{\sqrt{d \tr \Sigma}}{n}} \geqslant 1 - \delta
\end{equation}
The algorithm takes
\begin{itemize}
    \item $O\left(n d^\frac{1}{4} \log \frac{d}{\delta}\right)$ in terms accesses to the quantum experiment
    \item $O\left(1\right)$ in terms of quantum registers needed to hold quantum experiments.
\end{itemize}
\end{theorem}
\begin{proof}
    Follows directly from Theorem~\ref{thm: con_sim_estimator}, \ref{thm: con_sim_complexity}, and \ref{thm: generic_multi_main}.
\end{proof}

\begin{theorem}[Main Result for Meticulous Estimator]
\label{thm: mer_main}
Given quantum experiments for a multivariate random variable $\vec{\RV{X}}$ with unknown covariance matrix $\Sigma$, there is an efficient algorithm that outputs a mean estimate $\vec{\tilde{\mu}}$ such that:
\begin{equation}
    \prob{\norm{\vec{\tilde{\mu}} - \expect{\vec{\RV{X}}}}_\infty \leqslant \frac{\sqrt{\tr \Sigma}}{n}} \geqslant 1 - \delta
\end{equation}
which implies 
\begin{equation}
    \prob{\norm{\vec{\tilde{\mu}} - \expect{\vec{\RV{X}}}}_2 \leqslant \frac{\sqrt{d \tr \Sigma}}{n}} \geqslant 1 - \delta
\end{equation}
The algorithm takes
\begin{itemize}
    \item $O\left(n \log \frac{d}{\delta}\right)$ in terms accesses to the quantum experiment
    \item $O\left(\log n \log \log n\right)$ in terms of quantum registers needed to hold quantum experiments.
\end{itemize}
\end{theorem}
\begin{proof}
    Follows directly from Theorem~\ref{thm: con_mer_estimator}, \ref{thm: con_mer_complexity}, and \ref{thm: generic_multi_main}.
\end{proof}

\section{Discussion and Application}
\label{sec: discussion}

We have just completed mathematical derivations and found two separate multivaraite estimators. The Meticulous Estimator theoretically outperforms previous algorithms in Ref.~\cite{Cornelissen_2022}, eliminating almost all log factors. Meanwhile, the simple estimator offers a straight-forward solution in the regime where $\log n \gg d^\frac{1}{4}$ and is memory-efficient.  In this section we make detailed discussions about their implications and future prospects.

\subsection{Utilizing Quantum Monte Carlo Directly}

One of the noticeable feature of the entire logic chain presented in this work is the fact that it deviates from the original derivations of Ref.~\cite{Kothari_2022}. We proved a different, and arguably stronger property of the Grover operator in Theorem~\ref{thm: alpha_close} and start from here. The only part where the logic collides with Ref.~\cite{Kothari_2022} is the classical reduction for the univariate mean value estimator. Is this justified?

For the simple estimator which trades a quartic slowdown but for memory-efficiency, the answer is yes. Our property in Theorem~\ref{thm: alpha_close} shows closeness for the original values, as opposed to its counterpart in Ref.~\cite{Kothari_2022}, which only demonstrates closeness for the absolute values. This property allows us to bound the distance between states which is vital to allow for multidimensional phase estimation. 

For the merticulous estimator, the answer is mixed. In Sec.~\ref{sec: mer_estimator}, we practically provided a pipeline that takes in any univaraite estimator and spits out a multivariate estimator. The only restriction is that the univariate estimator must be a hybird (reversible) circuit, where we are not allowed to skip any part of the execution. In other words, we want the number of times we called the Grover operator in the phase estimation to be a constant, such that it is independent of each thread in the quantum parallelization in the final algorithm. The original work in Ref.\cite{Kothari_2022} does not have this property. Nevertheless, we might find ways to convert it into so by admitting some constant overhead. In that sense, the merticulous estimator can probably be constructed from Ref.~\cite{Kothari_2022} directly, with some more work. 

\subsection{The final $\log \frac{d}{\delta}$ factor}
\label{sec: log_1/2}
Even the meticulous estimator we found is still not quite optimal, as there is a $O(\log \frac{d}{\delta})$ factor in our complexity. It is probably more pleasant if we have instead something that goes like $O(\log \frac{1}{\delta})$, because then for a constant $\delta$ we would not suffer from an additional $O(\log d)$ for large $d$. Note that this factor essentially comes from the properties of multidimensional phase estimation in Algorithm~\ref{alg: phase_estimation}, as it output results with a constant success probability for each dimension separately. Coalescing these dimensions into the $\norm{}_\infty$ bound causes this extra factor which finds a way into our final result. 

To remove this factor, one may find ways to generalize the multidimensional phase estimation and potentially obtain something different. The nicest way is to think of the hypercubic lattice $G$ as the name suggests---a lattice in the real space. Then, multivariate phase estimation essentially tries to pinpoint the corresponding location of the unknown vector on the reciprocal lattice in momentum space. This allows us to generalize multidimension phase estimation expect $G$ doesn't have to be an orthogonal lattice. However, this still won't give us any progress.

To navigate the issue, one could define multivariate phase estimation except the coordinates in $G$ are somehow chosen from a sphere. However, this choice loses translational symmetry in the lattice, which leads to complications. At best, we might only be able to work with some form of spherical harmonics transform instead of Fourier transform, but the meaning of ``phase estimation'' remains unclear in this setting. Even if we can do such, there is no guarantee that we will work with something such as $\expval{\vec u, \vec{\RV{X}}}$ in the phase, which is a linear function of $\vec{\RV{X}}$. The linearity in the phase is important because we have $\expect{\expval{\vec u, \vec{\RV{X}}}} = \expval{\vec u, \expect{\vec{\RV{X}}}}$. If we relinquish such linearity it is unlikely that we get anything useful. For such reasons optimizing on this $\log \frac{d}{\delta}$ might be very difficult. In fact, we conjecture that our result is optimal for the $\norm{}_\infty$ norm. 

\subsection{Practical Considerations}

As discussed in Sec.~\ref{sec: primitives}, in practice, what we are presented is usually some program that generates a superposition of all possible situations representing all the possibilities of the value we are trying to measure, as we defined as ``complete quantum experiment'' in Definition~\ref{def: complete_quantum_exp}. In this paper we partitioned such a program into two parts, the synthesizer $\mathcal{P}$ which generates all the possibilities, and the quantum experiment which computes the value for each possibility. In practice, however, it may not be simple to create such partition. A program might create different possibilities along its execution and there is no clean distinction between these two stages. One easy solution is to note that we never truely define each instance $\ket{k}$ for $k \in \Omega$ means. So we can simply absorb each instance $\vec{\mathcal{X}}_k$ as a part of $\ket{k}$. In other words, given complete quantum experiment $V$ we set $U = I$ and $P = V$. That resolves the issue. 

By now it has been quite evident it really shouldn't matter how we partitioned the entire program, the complete quantum experiment whatsoever, as found in the following observation:
\begin{remark}
After connecting all the steps we discussed to build the entire mean value estimation algorithm, we can write it in a way that only involves $V = U \mathcal{P}$ instead of $U$ or $\mathcal{P}$ individually.
\label{remark: join}
\end{remark}
To be more exact, we can coalesce adjacent Grover operators in the phase estimation step,  such that our phase estimation algorithm always uses $U \mathcal{P}$ and its adjoint except in the beginning and the end. We then notice that we start with $\ket{\bm{1}} = \mathcal{P}\ket{0}$ so we insert $\mathcal{P}$ in the beginning and we can also insert $\mathcal{P^\dag}$ in the end since it doesn't affect our measurement outcome. Thus, the entire algorithm only utilizes $V = U \mathcal{P}$. 

Thus, we can see that the partition we introduced in the paper is arbitrary and does not matter whatsoever. We can, say, set $\mathcal{P} = I$ and $U = V$, then the phase estimation still returns the correct answer. 

\subsection{Generalizations}

In Sec.~\ref{sec: primitives} we briefly discussed some examples of the quantum experiment, such as QCD simulation or QML circuit, where a classical analogue is not necessary. The more interesting case, however, is when we are working with multiple observables and we are trying to estimate their expectation respect to some state $\ket{\psi}$ simultaneously. This has been explored in Ref.~\cite{Huggins_2022}, where the authors found a quadratic speedup up to polylog factors when the spectral norm of the observables $\norm{O}_j$ is bounded by some constant, or different constants. Is there an way that we can utilize the ideas developed in this paper to tackle this situation?

The answer is mixed. The key difference for general observables is that we have no knowledge of the eigenbasis for the observables and are therefore unable to efficiently perform the post-processing of the variables. For $\arctan$ function, we might be able do the transformations with block encoding and signal processing methods \cite{Dalzell_2023}. This would admit $\text{polylog}(n)$ factors in the algorithm. Furthermore, we are not aware of efficient methods to implement truncation in this situation. Without truncation, we can only do the bounded case, except the upper bound is no longer a constant but scales inversely with the true mean, similar to conditions in Theorem~\ref{thm: alpha_close}. If one could find a way to implement truncation and analyze the overhead to implement $\arctan$ transformation, then this slight potential advantage might be generalized to tackle the case of general observables.

\section{Conclusion}

In this paper, utilizing the idea from quantum Monte Carlo \cite{Kothari_2022}, we developed an efficient quantum multivariate mean value estimation algorithm that almost saturates the optimal bound discussed in Ref.~\cite{Cornelissen_2022}. We also discussed another memory-efficient estimator, which foundamentally originates from a stronger property of the Grover operator primitive compared to Quantum Monte Carlo. Compared to previous work \cite{Cornelissen_2022}, we eliminated (almost) all polylog factors that is feasible. We then discussed its applications and potential generalizations to the case of estimating expectation for generic observables.  

\section{Acknowledgement}

This problem was provided to me by Prof. Ryan O'Donnell at CMU who played a vital role during the development of this work. I thank him for valuable discussions, continued support, and important advice. 

\appendix

\section{Completing Theorem~\ref{thm: log_log}---the log-log trick}
\label{app: log_log}

In Theorem~\ref{thm: log_log} we discussed that we need a better assignment of parameters than those Ref.~\cite{Kothari_2022} to satisfy the additional constraints. Here we gave an explicit proof:
\begin{theorem}[Theorem~\ref{thm: log_log} Restated]
\label{thm: log_log_app}
Give some algorithm $\mathcal{A}$ depending on parameters $(\varepsilon, \delta)$, where $\varepsilon > 0$ and $0 < \delta < 1$, such that 
\begin{itemize}
    \item Algorithm always costs $O\left(\frac{1}{\varepsilon} \log \frac{1}{\delta} \right)$ by some measure of complexity and $O\left(\log \frac{1}{\delta}\right)$ by another;
    \item Algorithm ``succeed'' with probability at least $1 - \delta$ whenever it is called. ``succeed'' can be defined as some logical predicate.
\end{itemize}

Fix $\varepsilon, \delta$. Consider calling the algorithm $T$ times. Let the $j$-th time the algorithm to be called with $(\varepsilon'_j, \delta'_j)$. We fix $\varepsilon'_j$ such that $\varepsilon'_T = \varepsilon$ and $\varepsilon'_{j+1} \leqslant \frac{\varepsilon'_{j}}{R}$, where $R > 1$ is some fixed constant. By setting 
\begin{equation}
    \delta'_j = \frac{6}{\pi^2} \frac{1}{\left(T-j + 1\right)^2} \delta
\end{equation}
We can make sure that 
\begin{itemize}
    \item All calls to algorithm $\mathcal{A}$ succeed simultaneously with probability at least $1 - \delta$.
    \item The combined cost is $O\left(\frac{1}{\varepsilon} \log \frac{1}{\delta}\right)$ and $O\left(T \log \frac{T}{\delta}\right)$ by the two measures of complexity respectively. 
\end{itemize} 
\end{theorem}
\begin{proof}
First, by union bound we have
\begin{equation}
\prob{\text{any call fail}} \leqslant \sum_{j=1}^T \delta'_j = \frac{6}{\pi^2} \delta \sum_{j=1}^T \frac{1}{j^2} \leqslant \delta
\end{equation}
As $\sum_{j=1}^T \frac{1}{j^2} \leqslant \sum_{j=1}^\infty \frac{1}{j^2} \leqslant \zeta(2) = \frac{\pi^2}{6}$. So all calls succeed simultaneously with probability at least $1 - \delta$.

Then we try to bound the complexity. The first measure (in terms of big-$O$) gives:
\begin{equation}
\begin{aligned}
\sum_{j=1}^T \frac{1}{\varepsilon'_j} \ln \frac{1}{\delta'_j} & \leqslant \sum_{j=1}^T \frac{1}{R^{T-j} \varepsilon} \ln \left(\frac{\pi^2}{6} \left(T-j+1\right)^2 \frac{1}{\delta}\right) = \frac{1}{\varepsilon}\sum_{j=1}^T \frac{1}{R^{j-1}} \ln \left(\frac{\pi^2}{6} j^2 \frac{1}{\delta}\right) \\
& \leqslant \frac{1}{\varepsilon} \left( 
\sum_{j=1}^T \frac{1}{R^{j-1}} \ln \left(\frac{\pi^2}{6} \right)   + \sum_{j=1}^T \frac{2}{R^{j-1}} \ln j  + \sum_{j=1}^T \frac{1}{R^{j-1}} \ln \frac{1}{\delta} \right) \\
& \leqslant \frac{1}{\varepsilon} \left( 
\frac{R}{R-1} \ln \left(\frac{\pi^2}{6} \right)  + O(1) +  \frac{R}{R-1} \ln \frac{1}{\delta}\right) \\
& \leqslant O\left(\frac{1}{\varepsilon}\log \frac{1}{\delta}\right)
\end{aligned}
\end{equation}
It is trivial that $\sum_{j=1}^T \frac{2}{R^{j-1}} \ln j$ is a constant, as we can bound with $\sum_{j=1}^\infty \frac{2}{R^{j-1}} \ln j$ which converges. For the sake of completeness we will show the convergence explicitly. First, it is easy to see that starting from some $j = j_0$ the series must be monotonically decreasing. So we can use the integral test, where we note that:
\begin{equation}
\begin{aligned}
     \int_1^\infty \dd{x} \frac{2}{R^{x-1}} \ln x & = 2R \int_1^\infty  \dd{x} R^{- x} \ln x = \frac{2R}{\ln R} \int_1^\infty \dd{x} \frac{e^{-(\ln R) x}}{x} = \frac{2R}{\ln R} \int_{(\ln R)}^\infty \dd{x} \frac{e^{-x}}{x} 
     \\ & = \frac{2R}{\ln R} \left(-\operatorname{Ei}(-\ln R) \right) = \frac{2R}{\ln R} \Gamma(0,\ln R) \in O(1)
\end{aligned}
\end{equation}
Where we used integration by parts and substitution during the calculation. $\operatorname{Ei}$ is exponential integral function, and $\Gamma$ is the incomplete Gamma function. This completes that the total cost by the first measure of complexity we discussed is $O\left(\sum_{j=1}^T \frac{1}{\varepsilon'_j} \ln \frac{1}{\delta'_j}\right) = O\left(\frac{1}{\varepsilon}\log \frac{1}{\delta}\right)$.

For the other measure of complexity we find:
\begin{equation}
    \sum_{j=1}^T \ln \frac{1}{\delta'_j} = \sum_{j=1}^T \ln \left(\frac{\pi^2}{6} \left(T-j+1\right)^2 \frac{1}{\delta}\right) = \sum_{j=1}^T \ln \left(\frac{\pi^2}{6} j^2 \frac{1}{\delta}\right) =  
T \ln \left(\frac{\pi^2}{6} \right)   + \sum_{j=1}^T \ln j  + T \ln \frac{1}{\delta} \in O\left(T \log \frac{T}{\delta}\right)
\end{equation}
This completes the proof.
\end{proof}
As a comment, remember that the name ``log log'' trick comes from eliminating the extra $\log \log \frac{1}{\varepsilon}$ factor if we do not vary $\delta'_j$ for different calls. But here we see that if we consider both measures of complexity, what we are really doing is to trade the log-log factor in one measure of complexity to the other. (i.e, we will now suffer a log log factor for the other complexity, which will eventually leads to a log log factor in space complexity in terms of quantum registers).

\section{Continuation of Theorem~\ref{thm: spectrum_G}}
\label{app: eigen_continued}

In Theorem~\ref{thm: spectrum_G} we analyzed the spectrum of the Grover operator. However, in Sec.~\ref{sec: sectrum_Grover} we only proved the part of the theorem that is useful to our later analysis. Here we will complete the proof for the theorem. 
\begin{theorem}[Spectrum of Grover Operator, Restated]
\label{thm: spectrum_G_app}
The spectrum of the Grover operator $\mathcal{G}$ for a real univariate variable $\RV{\theta}$ contains all eigenvalues of the form $e^{i \alpha}$ where $-\pi < \alpha \leqslant \pi$ satisfies Eq.~(\ref{eq: alpha_eigen}) (assuming $\tanft{\theta_k - \alpha}$ do not blow up for any $k \in \Omega$). Corresponding to $\alpha$, the eigenvector, specified by $\ket{\RV{\psi}}$ where $\RV{\psi}$ is a complex random variable, satisfy Eq.~(\ref{eq: eigen_vecs}).

In the event that there are multiple outcomes in $\RV{\theta}$ that share the same value (mod $2\pi$). For each possible such value $\varphi \in (-\pi, \pi]$, let $S = \{k \in \Omega: \mathcal{\theta}_k = \varphi \pmod{2\pi}\}$. We find, in addition, an eigenvalue $e^{i\alpha}$ satisfying $\varphi - \alpha = \pi \pmod{2\pi}$, and eigenvectors $\ket{\RV{\psi}}$ (where $\RV{\psi}$ is a complex random variable) with 
\begin{equation}
\left\{
\begin{gathered}
    \psi_k = 0 \quad \forall k \not\in S \hfill \\
    \expect{\RV{\psi}} = 0 \hfill \\
\end{gathered}
\right.
\end{equation}
\end{theorem}
\begin{proof}
In the main text, we have simplified the equation for eigenvalues down to Eq.~(\ref{eq: simple_eigen_eq_2}), which we will restate here:
\begin{equation}
\label{eq: app_Grover_1}
\expect{\RV{\psi}} = \frac{1 + e ^ {i (\RV{\theta} - \alpha)}}{2} \RV{\psi}
\end{equation}
To complete our cases, now consider $\alpha$ such that there are some $k \in \Omega$ such that $e^{i \left(\theta_k - \alpha\right)} = -1$. Then we immediately finds $\expect{\RV{\psi}} = 0$. Let $A = \left\{ k \in \Omega: \; \theta_k-\alpha = \pi \pmod{2\pi} \right\}$. This means that $\psi_k = 0$ for all $k \not\in A$. Combined, that is:
\begin{equation}
\left\{
\begin{gathered}
    \psi_k = 0 \quad \forall k \not\in A \hfill \\
    \expect{\RV{\psi}} = 0 \hfill \\
\end{gathered}
\right.
\end{equation}
We can show that under the assumptions we made on $\alpha$, this equation is equivalent to Eq.~(\ref{eq: app_Grover_1}). Note that such equation has non-trivial solution if and only if $|A| \geqslant 2$. This means that we find multiple $\theta_k$ takes the same value ($\alpha + \pi$), mod $2\pi$, for all $k \in A$. Conversely, for each value $\varphi$ such that there are multiple $k \in \Omega$ with $\theta_k = \varphi$, setting $\alpha = \theta - \pi \pmod{2 \pi}$ immediately traces back to satisfy Eq.~(\ref{eq: app_Grover_1}) and shows corresponding eigenvalues and eigenkets. Combined with the main text, we complete the proof for Theorem~\ref{thm: spectrum_G}.

\end{proof}

\bibliography{references.bib}

\end{document}